\documentclass[final]{siamltex}

\usepackage{graphicx}

\usepackage{amssymb, amsmath}

\usepackage{bm}

\newcommand{\be}{\begin{equation}}
\newcommand{\ee}{\end{equation}}

\newcommand{\dd}{{\rm d}}
\newcommand{\e}[1]{\,{\rm e}^{#1}\,}

\newcommand{\bsa}{{\mathbf a}}
\newcommand{\bsb}{{\mathbf b}}
\newcommand{\bse}{{\mathbf e}}
\newcommand{\bsr}{{\mathbf r}}
\newcommand{\bsp}{{\mathbf p}}
\newcommand{\bsv}{{\mathbf v}}
\newcommand{\bsw}{{\mathbf w}}
\newcommand{\bsz}{{\mathbf z}}

\newcommand{\bsA}{{\mathbf A}}
\newcommand{\bsD}{{\mathbf D}}
\newcommand{\bsX}{{\mathbf X}}
\newcommand{\bsZ}{{\mathbf Z}}

\newcommand{\bGa}{{\boldsymbol \Gamma}}

\newcommand{\cF}{{\mathcal F}}
\newcommand{\cL}{{\mathcal L}}
\newcommand{\cN}{{\mathcal N}}
\newcommand{\cO}{{\mathcal O}}
\newcommand{\cP}{{\mathcal P}}

\newcommand{\bsone}{{\mathbf 1}}

\newcommand{\R}{\mathbb R}

\newcommand{\N}{\mathbb N}

\renewcommand{\grad}[1]{{\boldsymbol \nabla}_{#1}}
\renewcommand{\div}[1]{\boldsymbol \nabla_{#1} \cdot}

\usepackage{showkeys}

\usepackage{url}

\title{The overdamped limit of dynamic density functional theory: Rigorous
results\thanks{This work is supported by European Research Council
Advanced Grant No.\ 247031 and EPSRC Grant No.\ EP/H034587/1.}}

\author{B. D. Goddard\thanks{Department of Chemical Engineering, Imperial
College London, London SW7 2AZ, UK (b.goddard@imperial.ac.uk).} \and
G. A. Pavliotis\thanks{Department of Mathematics, Imperial College
London, London SW7 2AZ, UK (g.pavliotis@imperial.ac.uk).} \and S.
Kalliadasis\thanks{Department of Chemical Engineering, Imperial
College London, London SW7 2AZ, UK (s.kalliadasis@imperial.ac.uk).}}

\begin{document}

\maketitle

\begin{abstract}
Consider the overdamped limit for a system of interacting particles
in the presence of hydrodynamic interactions. For two-body
hydrodynamic interactions and one- and two-body potentials, a
Smoluchowski-type evolution equation is rigorously derived for the
one-particle distribution function. This new equation includes a
novel definition of the diffusion tensor.  A comparison with
existing formulations of dynamic density functional theory is also
made.
\end{abstract}

\begin{keywords}
dynamic density functional theory, colloids, overdamped limit, hydrodynamic
interactions, Hilbert expansion, Smoluchowski equation, homogenization.
\end{keywords}

\begin{AMS}
82C22, 76M45, 76M50, 35B40
\end{AMS}

\section{Introduction} \label{S:Introduction}


\subsection{Review of existing dynamic density functional theory}

Several problems in condensed matter physics such as colloidal
suspensions and polymers can often be described as systems of
interacting Brownian particles, either in the presence or absence of
hydrodynamic interactions~\cite{Dhont96,Ottinger96}. Hydrodynamic
interactions are due to forces on the colloid particles caused by
flows in the suspending fluid (referred to as the bath), which are
generated by the motion of the colloidal particles, and can be
thought of as generalised friction forces. Such systems of
interacting Brownian particles can be described either in phase
space, when both positions and momenta of the particles are taken
into account, or in configuration space, when inertial effects are
neglected and only the position of the Brownian particles is taken
into account. The evolution of the phase space distribution function
is described by the Kramers equation~\cite{Klein22, Kramers40}. On
the other hand, the evolution of the distribution function in
configuration space is governed by the Smoluchowski
equation~\cite{Smoluchowski15, Einstein05}.

The Smoluchowski equation can be derived from the Kramers equation
in the overdamped, i.e.\ high friction, limit. In this limit the
velocity of the particles thermalizes quickly, i.e.\ the velocity
distribution converges quickly to a Maxwellian and the momentum
variables can be eliminated through an appropriate adiabatic
elimination procedure. This procedure is now well understood, both
for a single Brownian particle as well as for systems of interacting
particles, and both in finite as well as infinite dimensions
\cite{CerraiFreidlin06a,CerraiFreidlin06b,Nelson67,Risken84,Gardiner85}.

Whilst the derivation of the Smoluchowski equation was already
discussed by both Klein and Kramers, their approach was largely
heuristic and a rigorous theory was not introduced until
later~\cite{Nelson67}. Systematic adiabatic elimination techniques
were introduced in the '70s and applied to the problem of the
rigorous derivation of the Smoluchowski equation for the
Kramers equation, e.g.\ by Wilemski \cite{Wilemski76} and Titulaer
\cite{Titulaer78,Titulaer80}. In particular, Titulaer considered the
fully-interacting $N$-body linear Kramers equation and used a
multiple-time-scale or Chapman-Enskog expansion in the friction
constant to systematically derive the $N$-body Smoluchowski equation
and its corrections. These systematic adiabatic elimination
procedures can be understood in the context of singular perturbation
theory for Markov processes \cite{Papanicolaou76} or, more
generally, in the framework of multiscale methods
(e.g.~\cite{PavliotisStuart08}). A pedagogical discussion of the
multiple-time-scale technique was given by Bocquet \cite{Bocquet97},
for the case of non-interacting particles and for the one-body
reduced distribution function. See also~\cite[Ch.
11]{PavliotisStuart08} and~\cite[Ch. 8]{pavliotis_lecture_notes}. It
is worth noting that in all these derivations, the equations for the
one-body distribution function are linear.


Both the Kramers and Smoluchowski equations describe the full
$N$-body dynamics, and, although they are linear, they are not
well-suited to computation. This is due to the large number of
variables, and whilst the derivation of the Smoluchowski equation
from the Kramers equation is of fundamental interest, the reduction
from $6N$ to $3N$ variables by eliminating the momentum variables is
insignificant in terms of computational complexity. However, a
further simplification arises by integrating out over the positions
(and momenta) of all but one particle, which then allows to obtain
the dynamics of the reduced distribution functions~\cite{Balescu97,
ResiboisDeLeener77}; see also (\ref{fnDefn}). One of the main goals
of statistical mechanics and kinetic theory is the derivation of
closed equations for these reduced distribution functions, and in
particular for the one-body distribution function. If there are no
inter-particle interactions, i.e.\ the only force comes from an
external potential, and hydrodynamic interactions between particles
are neglected, this reduction procedure results in one-particle
versions of the Kramers and Smoluchowski equations.  In the more
general case where such inter-particle effects may not be neglected,
the equations must be closed by choosing a suitable approximation of
the higher-body densities in terms of the one-body distribution
function, e.g.\ a mean field approximation. Such a description is
the ultimate aim of dynamic density functional theory (DDFT).

Consider a system of $N$ interacting particles with $N$-body
distribution function $f^{(N)}(\bsr_1, \bsp_1, \dots,
\bsr_N,\bsp_N,\tau)$, which gives the probability of finding
particles at $\bsr_1, \dots, \bsr_N$ with momenta $\bsp_1, \dots,
\bsp_N$ at time $\tau$. The derivation of a self-consistent DDFT
requires expressing the full $N$-body distribution function $f^{(N)}$ in terms
of the one-body reduced
distribution
$f^{(1)}(\bsr_1,\bsp_1,\tau)$ (see (\ref{fnDefn})) or, if starting
from the Smoluchowski equation, in terms of $\rho(\bsr_1,\tau):=\int
\dd \bsp_1 f^{(1)}(\bsr_1,\bsp_1,\tau)$. Whilst it is known that
$f^{(N)}$ (and thus all properties of the system including all lower
$n$-body distributions) is given by a unique functional of $\rho$,
both in \cite{HohenbergKohn64,Mermin65} and out \cite{ChanFinken05}
of equilibrium, in the general case this functional is unknown.
However, much work has been done for the equilibrium case (density
functional theory, or DFT which allows for an accurate description
of the microscopic properties of a fluid in terms of its density
distribution; see \cite{Evans79,RomanDietrich85} for early work and
e.g.\ \cite{Wu06,WuLi07} for recent overviews), and there exist
accurate functionals e.g.\ Rosenfeld's fundamental measure theory
for hard spheres
\cite{Rosenfeld89,RosenfeldSchmidtLowenTarazona97,RothEvansLangKahl02}
and the mean field approximation
\cite{LikosLangWatzalawekLowen01,Archer05}, mentioned earlier, which
becomes exact for soft interactions at high densities. DFT
represents one of the most widely used methods in condensed matter
physics for the study of the microscopic structure of
non-homogeneous fluids within the framework of equilibrium
statistical mechanics. It offers an increasingly popular compromise
between computationally costly molecular dynamics simulations and
various phenomenological approaches. It has been used to describe a
wide variety of physical settings, ranging from
polymers~\cite{Li05}, liquid crystals~\cite{deGennes93} and
molecular self-assembly~\cite{Tala03,Frink05} to interfacial
phenomena including wetting transitions on
substrates~\cite{Schick90,Berim09,Nold11}.

We now discuss how one moves from DFT to DDFT, where the system lies
away from equilibrium. DDFT is also a popular approach in condensed
matter physics, and has been applied to a wide range of problems
including spherical colloids without hydrodynamic interactions in
both configuration \cite{DietrichFrischMajhofer90,MarconiTarazona99,
MarconiTarazona00, EspanolLowen09} and phase space
\cite{MarconiMelchionna07,Archer09}, dense atomic liquids
\cite{Archer06}, anisotropic colloids \cite{RexWensinkLowen07}, and
inhomogeneous granular fluids \cite{MarconiTarazonaCecconi07}. The
effects of inertia \cite{MarconiTarazona06,MarconiTarazonaCecconiMelchionna07}
and
hydrodynamic interactions \cite{RexLowen09,Rauscher10} have also
been studied. However, none of the formalisms derived so far is
rigorous.  The relationships between different approaches are
summarized in Figure \ref{Fig:NonCommutative}.  We note, in
particular, that the two routes to obtain the one-body Smoluchowski
equation give, in general, different formulations. These derivations
can be divided into four cases, firstly by whether they start from
the Kramers or Smoluchowski equation, and secondly by whether or not
they include hydrodynamic interactions.

\begin{figure}
\begin{center}
\includegraphics[width=0.85\textheight, angle=90]{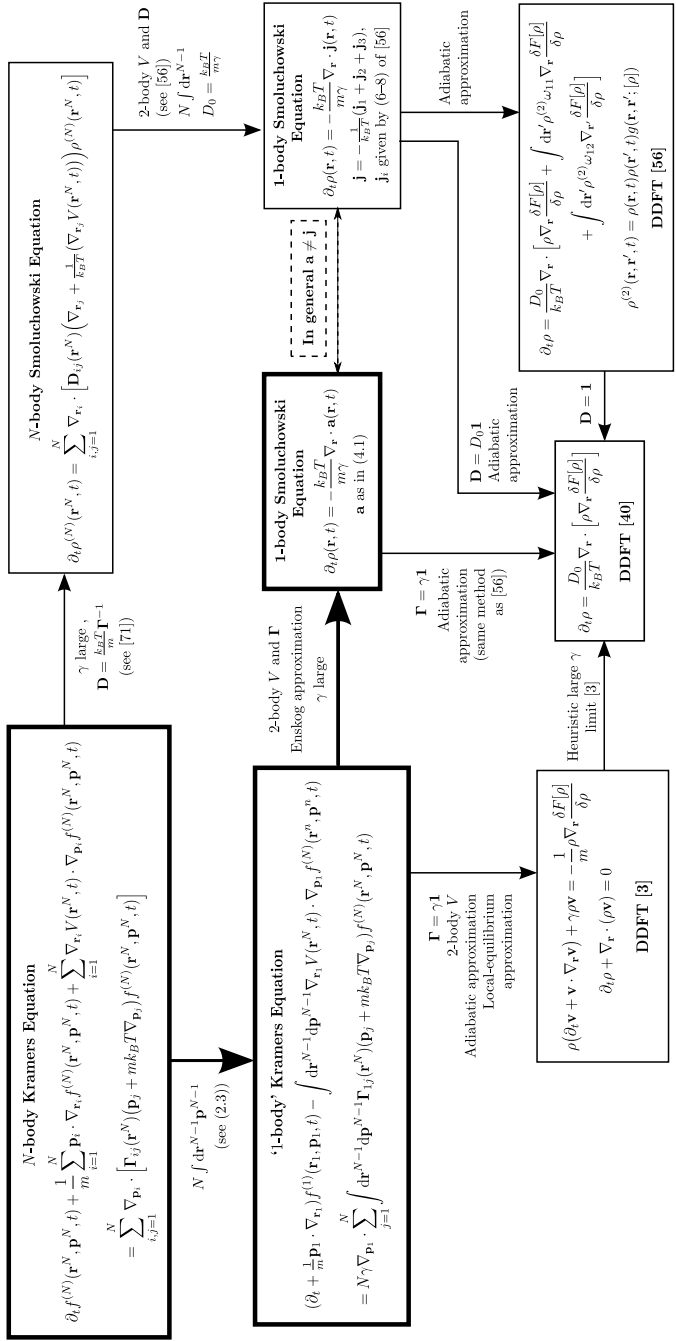}
\end{center}
\caption{Flow diagram of the various approaches used to obtain
one-body evolution equations and DDFTs from the full underlying
dynamics. Arrows indicate the interconnectedness of the different
approaches. Thick boxes/arrows: this work. Thin boxes/arrows:
previous approaches. Dashed boxes/arrows: note that the two routes
produce different one-body Smoluchowski equations when hydrodynamic
interactions are included, although both approaches are accurate to
$\cO(\epsilon^2)$. Text on arrows give brief descriptions of the
approximations made, see references for further details. Note in
particluar that the present formulation is a general one, and all
existing formulations may be derived from it.}
\label{Fig:NonCommutative}
\end{figure}


\subsection{Starting from the Kramers or the Smoluchowski equation}

As mentioned earlier, it is expected that the Smoluchowski equation
is valid in the overdamped limit, whereas for intermediate and small
values of the friction coefficient the Kramers equation should be
used. Perhaps the most common additional approximation is to ignore
the effects of the hydrodynamic interactions between the particles
(for exceptions, see \cite{RexLowen09,Rauscher10}). Whilst this
may be acceptable in a very dilute system, such interactions decay
only polynomially slowly with inter-particle distance, and are thus
long-range and important in many applications \cite{Dhont96}.


When starting from the Smoluchowski equation and neglecting
hydrodynamic interactions, it suffices to employ the adiabatic
approximation, first introduced by Marconi and Tarazona
\cite{MarconiTarazona99,MarconiTarazona00}. At equilibrium, Mermin's
proof~\cite{Mermin65} shows that there exists a unique functional of
$\rho$, $\cF_{\rm ex}[\rho]$, called the excess free energy
functional, which exactly determines the contributions from the
many-body potentials (which \emph{a priori} involve higher-order
reduced distributions). It then remains to determine accurate,
generally empirical, approximations to the unknown functional
$\cF_{\rm ex}[\rho]$.  The adiabatic approximation assumes that the
same relationship holds away from equilibrium.  This is equivalent
to assuming that the non-equilibrium $n$-body distributions are
identical to those in an equilibrium system with the same
instantaneous density $\rho$.  This approximation has proven
accurate in a range of systems
\cite{Archer05,ArcherEvans04,RexLowenLikos05,RoyallDzubiellaSchmidtBlaaderen07}.

If hydrodynamic interactions are included, this approximation is
insufficient.  This is because there are no hydrodynamic effects at
equilibrium.  Instead, at least for 2-body interactions, one uses
the identity
$\rho^{(2)}(\bsr,\bsr',\tau)=\rho(\bsr,t)\rho(\bsr',\tau)g(\bsr,\bsr';[\rho])$,
where $g$ is a pair-distribution function, while the
function $g-1$ is known as the pair correlation function (it
provides a measure of the distance over which particles are
correlated; for an ideal gas $g=1$) and assumes that a good
approximation to $g$ is known \cite{RexLowen09} (often $g$ can be
approximated with different methodologies, such as the BBGKY
hierarchy or the Ornstein-Zernike equation). Note in particular that
$g$ is a functional of $\rho$.


When starting from the Kramers equation an additional problem is
encountered. In this case one obtains an infinite hierarchy of
equations for the evolution of the momentum moments of
$f^{(1)}(\bsr,\bsp,\tau)$, i.e.\ for $\int \dd \bsp
p_1^{a_1}p_2^{a_2}p_3^{a_3} f^{(1)}(\bsr,\bsp,\tau)$ with $a_j \geq
0$, $\sum a_j=n$.  To obtain closure, one must truncate this
hierarchy at a given level, which requires the approximation of
higher moments. For example, the standard truncation at the velocity
($n=1$) level (i.e.\ the same level of description as the
Navier-Stokes equations), one must control terms (in appropriate
units) of the form $\int \dd \bsp \, (\bsp \otimes \bsp -\bsone)
f^{(1)}(\bsr,\bsp,\tau)$, where $\bsone$ is the 3$\times$3 identity
matrix.

At equilibrium, this term vanishes.  However, it is analogous to the
kinetic energy tensor in non-equilibrium thermodynamics
\cite{Kreuzer81} and thus is not negligible in general. Hence, for
atomic liquids \cite{Archer06}, it has been assumed that it can be
approximated by $\nu \partial_t \rho$, where $\nu$ is an arbitrary
collision frequency. Although this resulted in a DDFT analogous to
that previously derived for colloids in the high friction limit
\cite{MarconiTarazona99}, it is not clear that this is the correct
approximation in general.  For colloids with no hydrodynamic
interactions, this term can be dealt with using a local-equilibrium
approximation, or a Taylor expansion close to equilibrium, or
considering the high friction limit \cite{Archer09}. However, the
first two approaches are unsatisfactory for general systems which
may not lie close to (local) equilibrium, whilst the high friction
limit was not analysed rigorously.
This high friction limit is the main objective of the
present study. We will show that, in this limit, the term $\int \dd \bsp \,
(\bsp \otimes \bsp -\bsone) f^{(1)}(\bsr,\bsp,\tau)$ is indeed negligible
compared to $\rho$ and the momentum distribution.

It is worth noting that, if hydrodynamic interactions are neglected, the
heuristic high-friction calculation made by Archer \cite{Archer09} produces the
same DDFT as that derived by Marconi and Tarazona
\cite{MarconiTarazona99,MarconiTarazona00}. We shall demonstrate that this still
holds for the rigorous derivation.  However, when hydrodynamic interactions are
included, the two approaches do not lead to identical equations.  Section
\ref{S:Smoluchowski} discusses in detail these differences.


\subsection{Toward a rigorous derivation of dynamic density
functional theory}

Our main result is that, for a system of $N$ identical, spherically
symmetric colloid particles, up to errors of $\cO(\epsilon^2)$,
where $\epsilon \sim \gamma^{-1}$ with $\gamma$ the friction
constant for an infinitely dilute system (see Section \ref{S:Model}),
the dynamics of the one-body position distribution $\rho$ are given
by
\[
    \partial_\tau \rho(\bsr,\tau) = -\tfrac{k_BT}{m\gamma} \div{\bsr}
\bsa(\bsr,\tau),
\]
where $\bsa$ is the solution to a particular Fredholm integral
equation (Theorem \ref{T:MainTheorem}), $k_B$ is Boltzmann's
constant, $T$ the absolute temperature and $m$ the mass of the
colloid particles. Furthermore, we show that if the one-body phase
space distribution is written as a Hilbert \cite{Hilbert1953} (or
Champman-Enskog \cite{ChapmanCowling90}) expansion
$f^{(1)}(\bsr,\bsp,t)=f_0(\bsr,\bsp,t) + \epsilon f_1(\bsr,\bsp,t) +
\epsilon^2 f_2(\bsr,\bsp,t) + \cdots$ then (in appropriate units)
the first two terms are of the form
$\rho_0(\bsr,\tau)\e{-|\bsp|^2/2}$ and $\bsa(\bsr,\tau)\cdot\bsp
\e{-|\bsp|^2/2}$ respectively. In particular, non-zero terms in the
integral $\int \dd \bsp (\bsp \otimes \bsp - \bsone)
f^{(1)}(\bsr,\bsp,t)$ are at most $\cO(\epsilon^2)$.

We note that the evolution equation takes the form of a continuity
equation. In the framework of standard fluid dynamics one would
expect $\bsa(\bsr,t)=\rho(\bsr,t)\bsv(\bsr,t)$, where $\rho$ is the
fluid density and $\bsv$ is the velocity field.  Using the standard
definition \cite{Archer09} $\bsv(\bsr,t):=\rho^{-1}(\bsr,t)\int \dd
\bsp \, \bsp f^{(1)}(\bsr,\bsp,t)$, the Hilbert expansion
(\ref{fexpansion}), Corollary \ref{C:f0} and Lemma
\ref{L:f1explicit} show that, up to errors of $\cO(\epsilon^2)$,
this interpretation holds.


We now discuss the novelty of our approach and results.  In previous
work, the hydrodynamic interactions have been ignored. In this case,
the leading-order term in the expansion in the inverse of the
friction constant becomes linear (cf. (\ref{eps-2}) where it is
nonlinear), making the analysis significantly easier.  As noted
above, the full $N$-body equations are also linear. In this work, we
will consider two-body hydrodynamic interactions, along with a
two-body inter-particle potential, which require the approximation
of the two-body distribution.  As will be seen, a standard
approximation then leads to quadratic nonlinearities in the one-body
equation, formally analogous to the quadratic nonlinearity of the
collision operator in the Boltzmann equation
\cite{DeGrootMazur62,Cercignani75,ResiboisDeLeener77}. For more
general interactions, the nonlinearities will be of higher order; at
least \emph{a priori}, $n$-body interactions require $n$-body
distributions.  In Section \ref{S:Smoluchowski} we show that this
heuristic argument does not actually hold when starting from the
$N$-body Smoluchowski equation; higher distributions are required.

Let us now contrast the rigorous derivation of hydrodynamics from
the Boltzmann equation (in the limit of small mean free path, or
high collision frequency) -- see e.g. the comprehensive review by
Lebowitz \emph{et al} \cite{EspositoLebowitzMarra99}.  The approach
used therein, where time (and possibly space) are suitably rescaled
and a Hilbert expansion is used to derive an infinite hierarchy of
equations, which may then be solved to arbitrary order, is very
similar in spirit to ours. The collision term is replaced by a term
involving the hydrodynamic interactions, which has been much less
widely studied than the Boltzmann collision operator. Determining
the leading order term in the Hilbert expansion requires finding the
null space of the collision term (see (\ref{eps-2})). The full
friction operator is a complicated integral operator, and
determining its null space is non-trivial (see Lemma
\ref{L:NullSpaceL0+N0}).  In contrast, for the Boltzmann collision
and self-friction operators it is straightforward to show that the
null space contains only Maxwellians. Furthermore, in our situation,
there are additional nonlinear terms due to the inter-particle
potentials. However, due to these terms being independent of $\bsp$,
the momentum variable, and occurring with a higher power of the
small parameter, they do not hinder the analysis in the same way as
the hydrodynamic interaction terms. In addition, these
nonlinearities affect our ability to control the evolution of the
parts of the higher order corrections which lie in the null space of
the operator we need to invert. Sections \ref{S:epsilon1} and
\ref{S:Conclusions} highlight in detail these difficulties.


The structure of the paper is summarised as follows. In Section
\ref{S:Model} we give a description of the model, in both the
original and rescaled timescales, state our assumptions and give an
overview of the main result. In Section \ref{S:Hilbert} we develop
the solvability condition for the Hilbert expansion of the one-body
distribution $f^{(1)}$, which forms the basis for the proof of the
main result stated stated in Section \ref{S:Smoluchowski}, where we
also discuss its relationship with existing formulations of the
one-body Smoluchowski equation. In Section \ref{S:Conclusions} we
discuss the impact of our main result, including its application to
the derivation of DDFT, and also describe a number of associated
open problems. Appendix \ref{A:Proofs} contains proofs of the more
technical lemmas of Section \ref{S:Hilbert}.


\section{Description of the model and statement of main results} \label{S:Model}
We begin by considering the full equations of motion, in both
position and momentum for a large number $N$ of spherically
symmetric colloid particles of mass $m$ in a bath of a much larger
number of much lighter particles.  The interaction between the
colloidal particles and the bath is modelled on the level of
stochastic noise and the interaction between colloidal particles
mediated by the bath is modelled by friction terms.  The magnitude
of these two effects is correlated due to a generalised
fluctuation-dissipation theorem
\cite{Einstein05,MurphyAguirre72,Wilemski76,ErmakMcCammon78}.

The evolution equations are
\[
m \ddot \bsr_i = -\gamma m \sum_{j=1}^N \bGa_{ij} \dot \bsr_j + \bsX_i(\bsr^N) +
\sum_{j=1}^N \sqrt{2\gamma mk_BT} \bsA_{ij} \dot \bsw_j,
\]
where, for $\bsr^N=\bsr_1,\bsr_2,\dots,\bsr_N$, $\bGa_{ij} \in \R^{3\times
3}(\bsr^N)$ and $\bGa =(\bGa_{ij}) \in \R^{3N \times 3N}(\bsr^N)$ is the
friction tensor, which is positive definite, and in particular has a square
root.  $\gamma$ is the friction constant for a single isolated particle, and we
are interested in the regime where $\gamma \gg 1$.  The
$\dot\bsw_j(t)=(\dot w_j^1(t), \dot w_j^2(t), \dot w_j^3(t))^T$ are mean zero,
uncorrelated stochastic white noise terms and satisfy $\langle \dot
w_j^n(t)\rangle=0$ and $\langle \dot w_j^n(t) \dot
w_k^m(t')\rangle=\delta_{jk}\delta_{nm} \delta(t-t')$.  $\bGa$ is related to the
strength of the stochastic white noise terms $\dot \bsw_j$ via a generalized
fluctuation-dissipation theorem, namely $\bsA = \sqrt{\bGa}$. $\bsX_i$
represents the
force on particle $i$ exerted by an external field and interactions with the
other colloid particles.  $T$ is the absolute temperature, and $k_B$ is
Boltzmann's constant.

The motivation for this analysis is that in the high friction
(overdamped, large $\gamma$) limit, the momenta should reach
equilibrium on a much faster timescale than the positions. In
particular, we are interested in times of order $\cO(\gamma^{-1})$
and so begin by rescaling the time variable as $t= k_BT/(m\gamma)
\tau$ (where $t$ and $\tau$ are the new and old times respectively),
set $\bsX_i=-\grad{\bsr_i}U$ and define $\tilde
\bsX_i:=\bsX_i/(k_BT)=-\grad{\bsr_i}V=-\grad{\bsr_i}U/k_BT$ and
$\epsilon=\sqrt{k_BT/m}\gamma^{-1}$.  This rescaling leads to the
following system of equations:
\begin{align*}
\dot \bsr_i &= \frac{1}{\epsilon} \bsp_i \\
\dot \bsp_i &= -\frac{1}{\epsilon^2} \sum_{j=1}^N
\bGa_{ij}(\bsr^N) \bsp_j + \frac{1}{\epsilon} \tilde\bsX_i(\bsr^N) +
\sum_{j=1}^N \sqrt{2\epsilon^{-2}} \bsA_{ij} \dot \bsw_j.
\end{align*}

The constant in the time rescaling corresponds physically to $D_0$, the
diffusion constant for a single isolated particle.  The rescaling of $\bsX$
corresponds to measuring potential energy in units of the temperature.  For
$\epsilon$,
we note that $\sqrt{k_BT/m}$ is the average thermal equilibrium speed of a
particle at temperature $T$, whilst $\gamma^{-1}$ is approximately the time
required for the velocity distribution of the colloids to equilibrate.  Hence
$\epsilon$ has units of length, and in order to produce a non-dimensional
constant, it would be necessary to introduce another length scale.  Such a scale
is highly problem-dependent, and could for example be the typical length over
which the external potential varies, the length of a finite box in which the
particles are contained, a typical separation of colloid particles, or the size
of the colloids.  As such, we retain the dimensional parameter $\epsilon$ and
remark that the existence of a small parameter for applications must be checked
on a case-by-case basis.

The corresponding Fokker-Planck equation for the $N$-body distribution
function is
\begin{align}
&\partial_t f^{(N)}(\bsr^N,\bsp^N,t) +\frac{1}{\epsilon} \sum_{i=1}^N \bsp_i
\cdot \grad{\bsr_i} f^{(N)}(\bsr^N,\bsp^N,t) \notag\\
&\qquad -\frac{1}{\epsilon} \sum_{i=1}^N \grad{\bsr_i} V(\bsr^N) \cdot
\grad{\bsp_i}
f^{(N)}(\bsr^N,\bsp^N,t) \notag \\
&\qquad \qquad \qquad \qquad= \frac{1}{\epsilon^2} \sum_{i,j=1}^N \div{\bsp_i}
\Big[
\bGa_{ij}(\bsr^N) (\bsp_j + \grad{\bsp_j}) f^{(N)}(\bsr^N,\bsp^N,t) \Big].
\label{NBodyKramers}
\end{align}
Here we have used the notation $\bsr^n=(\bsr_1, \dots, \bsr_n)$ and the
analogue for $\bsp^n$.  We also find it convenient to write $\dd \bsr^{N-n}=\dd
\bsr_{n+1} \dots
\dd \bsr_N$ and the analogue for $\bsp$. In the
above, $f^{(N)}(\bsr^N,\bsp^N,t)$ is the probability of finding each particle
$i$ at position $\bsr_i$ with momentum $\bsp_i$ at time $t$.  We note that the
above equation is precisely the $N$-body Kramers equation, see e.g.\
\cite{DeutchOppenheim71,MurphyAguirre72}.

It is clear that $f^{(N)}$ encodes a huge amount of information, and
as such is very computationally demanding. We are not interested in
the distributions of the  positions and momenta of all the identical
particles, but in the distribution of the average values of these
quantities.  For this reason we introduce the reduced probability
distributions \be \label{fnDefn}
    f^{(n)}(\bsr^n,\bsp^n,t):= \frac{N!}{(N-n)!} \int
\dd \bsr^{N-n} \dd \bsp^{N-n} f^{(N)}(\bsr^N,\bsp^n,t).
\ee

Multiplying (\ref{NBodyKramers}) by $N$ and integrating over $\dd
\bsr^{N-1} \dd \bsp^{N-1}$, all terms in the sums with $i \neq 1$
vanish and the evolution of the one-body distribution is given by
\begin{align}
&\big( \partial_t + \frac{1}{\epsilon} \bsp_1 \cdot \grad{\bsr_1}
\big)
f^{(1)}(\bsr_1,\bsp_1,t)
-\frac{N}{\epsilon} \int \dd \bsr^{N-1}
\dd \bsp^{N-1} \grad{\bsr_1} V(\bsr^N) \cdot \grad{\bsp_1}
f^{(N)}(\bsr^N,\bsp^N,t) \notag \\
& \qquad \qquad =\frac{N}{\epsilon^2} \div{\bsp_1} \sum_{j=1}^N \int \dd
\bsr^{N-1}
\dd \bsp^{N-1} \bGa_{1j}(\bsr^N) (\bsp_j +
\grad{\bsp_j})
f^{(N)}(\bsr^N,\bsp^N,t).\label{oneBodyWithfN}
\end{align}

The difficulty in solving this equation lies primarily in the fact
that the last two terms still involve $f^{(N)}$, the full $N$-body
distribution.  In order to remove this dependence, and obtain a
closed equation, it is necessary to make some assumptions.  Firstly
we assume that the potential and friction tensor contain at most
two-body interactions. We will show that this is equivalent  to
requiring knowledge of only $f^{(2)}$.  From previous studies on the
derivation of DDFT (e.g.~\cite{ChanFinken05}), it is known that the
full $N$-body distribution function $f^{(N)}$ can be written as a
functional of the one-body spatial distribution, and therefore so
can $f^{(2)}$. We make the Enskog approximation to the two-body
distribution, in particular assuming that
$f^{(2)}(\bsr_1,\bsp_1,\bsr_2,\bsp_2,t)=f^{(1)}(\bsr_1,\bsp_1,t)f^{(1)}(\bsr_2,
\bsp_2,t) g(\bsr_1,\bsr_2)$, where the pair-distribution function
$g$ is assumed to be independent of $\bsp$ and $\epsilon$.

With this assumption, (\ref{fnDefn}) shows that $\int \dd \bsr \dd \bsp
f^{(1)}(\bsr,\bsp,t) g(\bsr,\bsr') = N-1$.  The role of $g$ is to describe
positional correlations of the particles, such as finite-size exclusion effects.
 It is an intermediate level of approximation between the mean field
approximation ($g \equiv 1$) and the full 2-body distribution
function $f^{(2)}$. We note that if $\bGa_{ij} = \delta_{ij}\bsone$, the
method outlined below allows a Smoluchowski equation to be derived
even if $g$ depends on $\bsp_1$ and $\bsp_2$.  It seems unlikely
than an analogous result holds in general, as
non-trivial momentum correlations on the two-particle level would
prevent the system from thermalising to a Maxwellian momentum
distribution.

The assumption that $g$ is independent of $\epsilon$ is known not to be valid
in general; we have only that $g=g(\bsr,\bsr';[\rho])$ ($g$ is a functional of
$\rho$), where $\rho(\bsr,t) = \int \dd \bsp f^{(1)}(\bsr,\bsp,t)$.  This
assumption does hold if the only contribution to $\rho$ comes from the zeroth
order term in $f^{(1)}$.  Even so, as we will discuss in Section
\ref{S:Conclusions}, if we expand $g$ in a power series in $\epsilon$, $g=g_0 +
\epsilon g_1 + \dots$, then the derivation changes only at the $\epsilon^0$
level, and we recover an analogous Smoluchowski equation. Such an approximation
is standard in the physics literature and, as suggested by the name, was first
proposed by Enskog \cite{Brush72} and revised by Van Beijeren and Ernst
\cite{VanBeijerenErnst73} to ensure consistency with irreversible
thermodynamics.  See also \cite{ResiboisDeLeener77}.

To summarize:
\begin{itemize}
\item {\bf Assumption 1} [Pairwise additive potential]
\[
    V(\bsr^N,t)=\sum_{i=1}^N V_1(\bsr_i,t) + \frac{1}{2}\sum_{i \neq j}
V_2(\bsr_i,\bsr_j).
\]
\item {\bf Assumption 2} [Pairwise additive friction]
\begin{align*}
    \bGa_{ij} = \delta_{ij} \bsone + \tilde \bGa_{ij} &= \delta_{ij} \big(
\bsone + \sum_{\ell \neq i} \bsZ_1(\bsr_i,\bsr_\ell) \big)
+ (1-\delta_{ij}) \bsZ_2(\bsr_i,\bsr_j)\\
&= \delta_{ij} \sum_{\ell \neq i} \big( \tfrac{1}{N-1}\bsone +
\bsZ_1(\bsr_i,\bsr_\ell) \big)
+ (1-\delta_{ij}) \bsZ_2(\bsr_i,\bsr_j),
\end{align*}
with the $\bsZ_j$ symmetric $3 \times 3$ matrices.
\\
\item {\bf Assumption 3} [Enskog approximation]
\be \label{Enskog}
 f^{(2)}(\bsr_1,\bsp_1,\bsr_2,\bsp_2,t)=f^{(1)}(\bsr_1,\bsp_1,t)f^{(1)}(\bsr_2,
\bsp_2,t) g(\bsr_1,\bsr_2).
\ee
\end{itemize}

Making these assumptions, we now calculate the two remaining terms in
(\ref{oneBodyWithfN}).  Using standard symmetry arguments gives:
\begin{align*}
&N\int \dd \bsr^{N-1}
\dd \bsp^{N-1} \grad{\bsr_1} V(\bsr^N) \cdot \grad{\bsp_1}
f^{(N)}(\bsr^N,\bsp^N,t) \\
&=\Big[ \grad{\bsr_1} V_1(\bsr_1,t)
+ \int \dd \bsr_2 \dd \bsp_2 g(\bsr_1,\bsr_2)
f^{(1)}(\bsr_2,\bsp_2,t) \grad{\bsr_1} V_2(\bsr_1,\bsr_2) \Big]
\cdot \grad{\bsp_1} f^{(1)}(\bsr_1,\bsp_1,t).
\end{align*}
and
\begin{align*}
& N\div{\bsp_1} \sum_{j=1}^N \int \dd
\bsr^{N-1} \dd \bsp^{N-1} \bGa_{1j}(\bsr^N) (\bsp_j +
\grad{\bsp_j}) f^{(N)}(\bsr^N,\bsp^N,t) \\
 &= \div{\bsp_1} \Big[ (\bsp_1 + \grad{\bsp_1}) f^{(1)}(\bsr_1,\bsp_1,t) \\
 & \qquad \qquad \quad + \int \dd \bsr_2 \dd \bsp_2
 g(\bsr_1,\bsr_2)f^{(1)}(\bsr_2,\bsp_2,t) \bsZ_1(\bsr_1,\bsr_2)  \times (\bsp_1
+ \grad{\bsp_1}) f^{(1)}(\bsr_1,\bsp_1,t) \\
& \qquad \qquad  \quad +\int \dd \bsr_2 \dd \bsp_2 g(\bsr_1,\bsr_2)
\bsZ_2(\bsr_1,\bsr_2) (\bsp_2 + \grad{\bsp_2}) f^{(1)}(\bsr_2,\bsp_2,t)
\times f^{(1)}(\bsr_1,\bsp_1,t) \Big].\\
\end{align*}

Hence, we have:
\begin{proposition}
Under Assumptions 1, 2 and 3, the evolution of the one-body reduced distribution satisfies
\begin{align}
 \partial_t f^{(1)}(\bsr,\bsp,t) &= \frac{1}{\epsilon} \Big[ -\bsp\cdot
\grad{\bsr} + \grad{\bsr}V_1(\bsr,t) \cdot \grad{\bsp} \notag \\
&\qquad + \int \dd \bsr' \dd
\bsp'
f^{(1)}(\bsr',\bsp',t) g(\bsr,\bsr') \grad{\bsr} V_2(\bsr,\bsr') \cdot
\grad{\bsp} \Big] f^{(1)}(\bsr,\bsp,t) \notag\\
&\hspace*{-12mm} + \frac{1}{\epsilon^2} \div{\bsp} \Big[ (\bsp +
\grad{\bsp}) f^{(1)}(\bsr,\bsp,t) \notag \\
& \hspace*{-12mm} \qquad \qquad \qquad  + \int \dd \bsr' \dd \bsp'
g(\bsr,\bsr') \bsZ_1(\bsr,\bsr') f^{(1)}(\bsr',\bsp',t) \times (\bsp +
\grad{\bsp}) f^{(1)}(\bsr,\bsp,t) \notag \\
&\hspace*{-12mm}  \qquad \qquad  \qquad  +\int \dd \bsr' \dd \bsp' g(\bsr,\bsr')
\bsZ_2(\bsr,\bsr') (\bsp' +  \grad{\bsp'}) f^{(1)}(\bsr',\bsp',t) \times
f^{(1)}(\bsr,\bsp,t) \Big] \notag \\
&=: \frac{1}{\epsilon}\big[\cL_1 f^{(1)} + \cN_1\big(
f^{(1)}, f^{(1)} \big)
\big] + \frac{1}{\epsilon^2}\big[\cL_0 f^{(1)} + \cN_0\big( f^{(1)}, f^{(1)}
\big) \big], \label{dtf1abstract}
\end{align}
where
\begin{subequations}
\begin{align}
    \cL_0 f &= \div{\bsp}(\bsp + \grad{\bsp}) f(\bsr,\bsp,t)
\label{L0}\\
    \cL_1 f &= [-\bsp \cdot \grad{\bsr} +
\grad{\bsr}V_1(\bsr,t)
\cdot \grad{\bsp}] f(\bsr,\bsp,t) \label{L1}\\
    \cN_0\big( f,\tilde f\big) &= \div{\bsp} \int \dd
\bsr' \dd \bsp' g(\bsr,\bsr') \bsZ_1(\bsr,\bsr') f (\bsr',\bsp',t) \times (\bsp
+ \grad{\bsp}) \tilde f(\bsr,\bsp,t)  \label{N0}\\
& \qquad  + \div{\bsp} \int \dd \bsr' \dd \bsp' g(\bsr,\bsr')
\bsZ_2(\bsr,\bsr') (\bsp' + \grad{\bsp'}) f(\bsr',\bsp',t) \times
\tilde f(\bsr,\bsp,t) \notag \\
    \cN_1 \big( f,\tilde f \big) &=  \int \dd \bsr'
\dd \bsp' f(\bsr',\bsp',t) g(\bsr,\bsr') \grad{\bsr} V_2(\bsr,\bsr') \cdot
\grad{\bsp} \tilde f(\bsr,\bsp,t), \label{N1}
\end{align}
\end{subequations}
and for ease of notation we have omitted the $(\bsr,\bsp,t)$ dependence of
the functions on the left hand sides of (\ref{L0})--(\ref{N1}).
\end{proposition}

It is noteworthy that, although they are quadratic, the non-linear
terms are not symmetric in the two arguments.  The first argument
has been taken to be that inside the integral.

In the following we will assume that $f^{(1)}$ is bounded and
positive\footnote{$f^{(1)}$ is non-negative by definition, but may
be zero, e.g.\ if there are excluded areas of the phase space due to
confining potentials.  It is clear that it may be made positive with
arbitrarily small errors.} and that all functions are sufficiently
regular and have sufficient decay at infinity for operators and
integrals to be defined.

To state our main result, we recall that the position
distribution, which is the object of interest in the single-particle
Smoluchowski regime, is defined by $\rho(\bsr,t)=\int\dd \bsp
f^{(1)}(\bsr,\bsp,t)$.  We will show that its evolution equation is
given by
\begin{theorem}[Smoluchowski equation] Under
suitable assumptions on $f^{(1)}$, $U_j$, $\bsZ_j$, $j=1,2$ (see
Theorem \ref{T:MainTheorem}), up to errors of
$\mathcal{O}(\epsilon^2)$ the dynamics of the one-body position
distribution are given (in the original timescale) by:
\[
     \partial_\tau \rho(\bsr,\tau)= -\tfrac{k_BT}{m\gamma}
\div{\bsr}\bsa(\bsr,\tau),
\]
where $\bsa(\bsr,\tau)$ is the solution to
\begin{align}
     & \bsa(\bsr,\tau) + \int \dd \bsr' g(\bsr,\bsr') \rho(\bsr',\tau)
\bsZ_1(\bsr,\bsr')\times \bsa(\bsr,\tau) + \rho(\bsr,\tau) \int \dd \bsr'
g(\bsr,\bsr')
\bsZ_2(\bsr,\bsr') \bsa(\bsr',\tau) \notag \\
&\qquad  = -\Big[ \grad{\bsr} + \tfrac{1}{k_BT} \Big(
\grad{\bsr}U_1(\bsr,\tau) + \int \dd \bsr' \rho(\bsr',\tau) g(\bsr,\bsr')
\grad{\bsr}U_2(\bsr,\bsr') \Big) \Big] \rho(\bsr,\tau). 
\end{align}
\end{theorem}


\section{The Hilbert Expansion} \label{S:Hilbert}
We now expand $f^{(1)}$ in powers of $\epsilon$ as:
\be \label{fexpansion}
    f^{(1)}(\bsr,\bsp,t)=\sum_{n=0}^\infty \epsilon^n f_n(\bsr,\bsp,t).
\ee Due to the singular nature of the problem, we do not expect such
a regular perturbation expansion to converge uniformly. The
expansion should be valid only for times $t \gg \epsilon$, and not
for shorter times, i.e.\ we expect there to be a boundary layer in
time of size $\cO(\epsilon)$. Since we are interested in times much
larger than $\epsilon$, interest lies in the leading order terms;
one would then hope to be able to truncate the series and prove
suitable bounds on the remainder term, as in
\cite{EspositoLebowitzMarra99}. We also assume that such an
expansion then converges, in particular that the $f_n$ are
sufficiently well-behaved in $\bsr$ and $\bsp$.

Inserting (\ref{fexpansion}) into the evolution equation (\ref{dtf1abstract})
and collecting powers of $\epsilon$ gives the hierarchy of equations
\begin{subequations}
\begin{align}
    \cL_0 f_0 + \cN_0(f_0,f_0) & =0  \label{eps-2}\\
    \cL_0 f_1 + \cL_1 f_0 + \cN_0(f_0,f_1) + \cN_0(f_1,f_0)  +
    \cN_1(f_0,f_0)&=0 \label{eps-1}\\
    \cL_0 f_2 + \cL_1 f_1 + \cN_0(f_2,f_0) + \cN_0(f_0,f_2) +
\cN_0(f_1,f_1) \qquad & \notag\\
     + \cN_1(f_1,f_0) + \cN_1(f_0,f_1) &= \partial_t f_0
    \label{eps0}\\
    \cL_0 f_n + \cL_1 f_{n-1} + \sum_{i+j=n} \cN_0(f_i,f_j) +\sum_{i+j=n-1}
\cN_1(f_i,f_j) &= \partial_t f_{n-2}, \; n\geq 3 \notag
\end{align}
\end{subequations}

We now solve these equations order-by-order.  First, to solve (\ref{eps-2}), we
need to determine the null space of $\cL_0 \cdot + \cN_0(\cdot,\cdot)$.
However, before we do so, the following lemma will be useful.  It is
essentially a result of the positive-definiteness of $\bGa$:
\begin{lemma} \label{L:IntegralPD}
For $\bsv(\bsr,\bsp,t)$ an arbitrary vector such that the integrals below exist,
and $f$ satisfying (\ref{Enskog}), there exists $\delta>0$ such that
\begin{align*}
    &\int \dd \bsr \dd \bsp \dd \bsr' \dd \bsp' f(\bsr,\bsp,t)
f(\bsr',\bsp',t) g(\bsr,\bsr') \\
& \qquad \qquad\times \big[ \bsv(\bsr,\bsp,t) \cdot \big( \tfrac{1}{N-1}
\bsone + \bsZ_1(\bsr,\bsr') \big) \bsv(\bsr,\bsp,t) + \bsv(\bsr,\bsp,t)
\cdot \bsZ_2(\bsr,\bsr') \bsv(\bsr',\bsp',t) \big]\\
&\qquad \geq \delta \int \dd \bsr \dd \bsp f(\bsr,\bsp,t)
|\bsv(\bsr,\bsp,t)|^2.
\end{align*}
In particular, the result holds when $f$ is chosen to be either $f_0$ or
$f^{(1)}$.
\end{lemma}
\begin{proof} See Appendix \ref{A:IntegralPD} \end{proof}

We now proceed with the analysis of (\ref{eps-2})--(\ref{eps0}), beginning by
determining the solution to (\ref{eps-2}):
\subsection{Solution of the $\epsilon^{-2}$ equation}
In this section we find the solution $f_0$ of (\ref{eps-2}):
\begin{lemma} \label{L:NullSpaceL0+N0}
For $f$ satisfying (\ref{Enskog}), the null space of $\cL_0 f + \cN_0( f, f)$
consists of functions of the
form $f(\bsr,\bsp,t)=\exp\big({-}|\bsp|^2/2 \big) \phi(\bsr,t)$.
\end{lemma}
\begin{proof}
We begin by assuming that $f$ satisfies (\ref{Enskog}) and
$\cL_0 f + \cN_0( f, f)=0$.  We define $\phi$ by $f(\bsr,\bsp,t)=:
\e{-|\bsp|^2/2} \phi(\bsr,\bsp,t)$ and note that $\phi$ is positive.
We therefore have $\ln \phi [ \cL_0 f + \cN_0( f, f)]=0$. Note that
$f(\bsr,\bsp,t) = \tfrac{1}{N-1} \int \dd
\bsr' \dd \bsp' f^{(2)}(\bsr,\bsp,\bsr',\bsp',t)$ and so
\[
    \cL_0 f(\bsr,\bsp,t) = \tfrac{1}{N-1} \div{\bsp}
\int \dd \bsr' \dd \bsp' g(\bsr,\bsr') f(\bsr',\bsp',t)(\bsp +
\grad{\bsp})
f(\bsr,\bsp,t),
\]
Using this reformulation along with the definition (\ref{N0}), the
notation $\tilde \bsZ_1(\bsr,\bsr'):=\tfrac{1}{N-1} \bsone +
\bsZ_1(\bsr,\bsr')$ and integrating over $\bsr$ and $\bsp$ gives
\begin{align*}
0&= \int \dd \bsr \dd \bsp  \ln \phi(\bsr,\bsp,t)
\div{\bsp} \int \dd \bsr'
\dd \bsp' g(\bsr,\bsr')  \Big[ f(\bsr',\bsp') \tilde \bsZ_1(\bsr,\bsr')
(\bsp +\grad{\bsp}) f(\bsr,\bsp,t)\\
& \qquad \qquad \qquad \qquad + f(\bsr,\bsp,t) \bsZ_2(\bsr,\bsr') (\bsp' +
\grad{\bsp'}) f(\bsr',\bsp',t) \Big]\\
&= -\int \dd \bsr \dd \bsp \dd \bsr'
\dd \bsp'  \e{-|\bsp|^2/2} \e{-|\bsp'|^2/2} g(\bsr,\bsr')
\frac{\grad{\bsp}\phi(\bsr,\bsp,t)}{\phi(\bsr,\bsp,t)
} \cdot \\
& \qquad \qquad \qquad \times \Big[ \phi(\bsr',\bsp',t) \tilde
\bsZ_1(\bsr,\bsr')
\grad{\bsp} \phi(\bsr,\bsp,t) +\phi(\bsr,\bsp,t) \bsZ_2(\bsr,\bsr') \grad{\bsp'}
\phi(\bsr',\bsp',t) \Big]
\end{align*}
where we have used integration by parts and Fubini's theorem, along with the
identity $(\bsp + \grad{\bsp})[\e{-|\bsp|^2/2} \phi(\bsr,\bsp,t)] =
\e{-|\bsp|^2/2} \grad{\bsp} \phi(\bsr,\bsp,t)$.

Letting $\bsv(\bsr,\bsp,t):=\grad{\bsp}\phi(\bsr,\bsp,t)/\phi(\bsr,\bsp,t)$ we
have
\begin{align*}
0&= -\int \dd \bsr \dd \bsp \dd \bsr'
\dd \bsp'  \e{-|\bsp|^2/2} \e{-|\bsp'|^2/2}
\phi(\bsr,\bsp,t) \phi(\bsr',\bsp',t) g(\bsr,\bsr')
\\
& \qquad \qquad \qquad \times \Big[ \bsv(\bsr,\bsp,t) \cdot
\tilde\bsZ_1(\bsr,\bsr') \bsv(\bsr,\bsp,t) + \bsv(\bsr,\bsp,t)
\cdot \bsZ_2(\bsr,\bsr') \bsv(\bsr',\bsp',t) \Big].
\end{align*}
Since $f(\bsr,\bsp,t)=\exp\big({-}|\bsp|^2/2\big) \phi(\bsr,\bsp,t)>0$, we may
apply Lemma \ref{L:IntegralPD}, which shows that
\[
    0= \int \dd \bsr \dd \bsp f(\bsr,\bsp,t) |\bsv(\bsr,\bsp,t)|^2 =
    \int \dd \bsr \dd \bsp \e{-|\bsp|^2/2} \phi^{-1}(\bsr,\bsp,t)
|\grad{\bsp}\phi(\bsr,\bsp,t)|^2.
\]
Hence, since $\phi$ is bounded, the integrand is zero if and only if
$\grad{\bsp}\phi(\bsr,\bsp,t)\equiv 0$ and the result holds.
\end{proof}

This result gives an explicit form for the $\bsp$-dependence of $f_0$.
\begin{corollary}
\label{C:f0}
The zeroth order term in the $\epsilon$-expansion of $f^{(1)}$ is given by
\[
    f_0(\bsr,\bsp,t)=\frac{1}{(2 \pi)^{3/2}}\e{-\frac{|\bsp|^2}{2}}
\rho_0(\bsr,t)=:Z^{-1} \e{-\frac{|\bsp|^2}{2}}
\rho_0(\bsr,t).
\]
\end{corollary}
\begin{proof}
Follows immediately from (\ref{eps-2}) and Lemma \ref{L:NullSpaceL0+N0}.
\end{proof}

\subsection{Solution of the $\epsilon^{-1}$ equation}
In order to find $f_1$ from (\ref{eps-1}), we rewrite it as
\be
\cL_0 f_1 + \cN_0(f_0,f_1) + \cN_0(f_1,f_0) = - \cL_1 f_0
    -\cN_1(f_0,f_0). \label{eps-1a}
\ee The first point to note is that for known $f_0$,  this is a
linear-operator equation for $f_1$.  Although this can be
seen in an abstract sense from the non-linearities being of a
quadratic nature, we will require the explicit form of the operator:
\begin{lemma} \label{L:N0}
For $\cN_0$ as in (\ref{N0}), $f_0(\bsr,\bsp,t)$ as given by Corollary
\ref{C:f0}, and arbitrary $\tilde f(\bsr,\bsp,t)$,
\begin{align*}
    \cN_0(f_0, \tilde f) &= \div{\bsp} \Big[ \int \dd
\bsr' \dd \bsp' g(\bsr,\bsr') f_0 (\bsr',\bsp',t) \bsZ_1(\bsr,\bsr')  \times
(\bsp + \grad{\bsp}) \tilde f(\bsr,\bsp,t) \Big] \\
    \cN_0(\tilde f, f_0) &= - f_0(\bsr,\bsp,t) \bsp \cdot \int \dd \bsr'
\dd \bsp' g(\bsr,\bsr') \bsZ_2(\bsr,\bsr') (\bsp' + \grad{\bsp'})
\tilde f(\bsr',\bsp',t).
\end{align*}
\end{lemma}
\begin{proof}
For $\cN_0(f_0, \tilde f)$ note that $(\bsp + \grad{\bsp})
\e{-|\bsp|^2/2}=0$ and hence the second line in (\ref{N0}) gives
zero and the first result follows. For $\cN_0(\tilde f, f_0)$, the
same argument shows that the first line in (\ref{N0}) is zero.  For
the remaining term, the result of the integral is a vector which
depends only on $\bsr$, and for any such vector $\bsz$,
$\div{\bsp}[\bsz(\bsr) f(\bsr,\bsp,t)]= \bsz(\bsr) \cdot
\grad{\bsp}f(\bsr,\bsp,t)$.  The result then follows from the
identity $\grad{\bsp} f_0(\bsr,\bsp,t) = -\bsp f_0(\bsr,\bsp,t)$.
\end{proof}
\begin{corollary}
For known $f_0$, and arbitrary $f$, $\tilde \cL f:=\cL_0 f +
\cN_0(f_0,f) + \cN_0(f,f_0)$ is a linear operator on $f$.
\end{corollary}

As is customary for problems of this form, it is convenient to work
in the $L^2$ space weighted by the inverse of the invariant measure
of $\tilde \cL$. In this case, the inner product is defined as
\[
\langle f(\bsr,\bsp,t),\tilde
f(\bsr,\bsp,t)\rangle_{f_0^{-1}}:= \int \dd
\bsr \dd \bsp f_0^{-1}(\bsr,\bsp,t) f(\bsr,\bsp,t) \tilde f(\bsr,\bsp,t).
\]
We denote this space by $L^2_{f_0^{-1}}$.  In this weighted space,
$\tilde \cL$ is self-adjoint and has compact resolvent, allowing us
to apply Fredholm's theory:
\begin{lemma} \label{L:SelfAdjoint}
$\tilde \cL$ is self-adjoint in $L^2_{f_0^{-1}}$.
\end{lemma}
\begin{proof} See Appendix \ref{A:SelfAdjoint} \end{proof}
\begin{lemma} \label{L:resolvent}
The resolvent of $\tilde \cL$ is compact in $L^2_{f_0^{-1}}$.
\end{lemma}
\begin{proof} See Appendix \ref{A:resolvent} \end{proof}

Since (\ref{eps-1a}) may be rewritten as $\tilde \cL f_1 = -\tilde \cL_1 f_0 -
\cN_1(f_0,f_0)$, to determine the solvability condition we must
determine the null space of $\tilde \cL^*$.
\begin{lemma}\label{L:nullLtildestar}
    The null space of $\tilde \cL$ (and of $\tilde \cL^*$) contains only
functions of the form $f(\bsr,\bsp,t)=\exp\big({-}|\bsp|^2/2 \big)
\phi(\bsr,t)$.
\end{lemma}
\begin{proof}
The proof is analogous to that of Lemma \ref{L:NullSpaceL0+N0}, but
is more straightforward as the operator is now linear and we are
already working in the appropriate weighted space.  As in the proof
of Lemma \ref{L:NullSpaceL0+N0}, we have
\[
    \cL_0 f(\bsr,\bsp,t) = \tfrac{1}{N-1} \div{\bsp}
\int \dd \bsr' \dd \bsp' g(\bsr,\bsr') f_0(\bsr',\bsp',t)(\bsp +
\grad{\bsp}) f(\bsr,\bsp,t),
\]
and hence, setting $\tilde \bsZ_1(\bsr,\bsr'):=\tfrac{1}{N-1}\bsone +
\bsZ_1(\bsr,\bsr')$ we have
\begin{align*}
    \tilde \cL f(\bsr,\bsp,t) &= \div{\bsp} \int \dd \bsr' \dd \bsp'
    g(\bsr,\bsr') f_0(\bsr',\bsp',t) \tilde \bsZ_1(\bsr,\bsr') (\bsp +
\grad{\bsp}) f(\bsr,\bsp,t)\\
 &\qquad \qquad \qquad \qquad + g(\bsr,\bsr') f_0(\bsr,\bsp,t)
\bsZ_2(\bsr,\bsr') (\bsp' + \grad{\bsp'}) f(\bsr',\bsp',t) =0.
\end{align*}

Taking the $L^2_{f_0^{-1}}$ inner product with $f$ gives
\begin{align*}
0 &=   \int \dd \bsr \dd \bsp f_0^{-1}(\bsr,\bsp,t) f(\bsr,\bsp,t) \\
& \qquad \times \div{\bsp}
\int \dd \bsr' \dd \bsp' \Big[
    g(\bsr,\bsr') f_0(\bsr',\bsp',t) \tilde \bsZ_1(\bsr,\bsr') (\bsp +
\grad{\bsp}) f(\bsr,\bsp,t)\\
 &\qquad \qquad \qquad \qquad \qquad \qquad  +
g(\bsr,\bsr') f_0(\bsr,\bsp,t)
\bsZ_2(\bsr,\bsr') (\bsp' + \grad{\bsp'}) f(\bsr',\bsp',t) \Big].
\end{align*}
Integrating by parts, using the identity $\grad{\bsp}f_0^{-1}(\bsr,\bsp,t) =
\bsp f_0^{-1}(\bsr,\bsp,t)$ and Fubini's theorem gives
\begin{align*}
0 &=   \int \dd \bsr \dd \bsp \dd \bsr' \dd \bsp' f_0^{-1}(\bsr,\bsp,t)
(\bsp + \grad{\bsp})f(\bsr,\bsp,t) \\
& \qquad \qquad \qquad \cdot \Big[ g(\bsr,\bsr') f_0(\bsr',\bsp',t)
\tilde \bsZ_1(\bsr,\bsr') (\bsp + \grad{\bsp}) f(\bsr,\bsp,t)\\
 &\qquad \qquad \qquad \qquad \qquad \qquad +
g(\bsr,\bsr') f_0(\bsr,\bsp,t)
\bsZ_2(\bsr,\bsr') (\bsp' + \grad{\bsp'}) f(\bsr',\bsp',t) \Big].
\end{align*}
Setting
$\bsv(\bsr,\bsp,t):=f_0^{-1}(\bsr,\bsp,t) (\bsp+\grad{\bsp})f(\bsr,\bsp,t) $
then gives
\begin{align*}
0=\int \dd \bsr \dd \bsp \dd \bsr' \dd \bsp'
f_0(\bsr,\bsp,t)&f_0(\bsr',\bsp',t) g(\bsr,\bsr') \bsv(\bsr,\bsp,t)\\
&  \cdot
\big[\tilde \bsZ_1(\bsr,\bsr') \bsv(\bsr,\bsp,t) +
\bsZ_2(\bsr,\bsr')\bsv(\bsr',\bsp')\big].
\end{align*}
Hence, by Lemma \ref{L:IntegralPD}, we have
\[
    0 = \int \dd \bsr \dd \bsp f_0(\bsr,\bsp,t) |\bsv(\bsr,\bsp,t)|^2
    = \int \dd \bsr \dd \bsp f_0^{-1}(\bsr,\bsp,t) |(\bsp +
\grad{\bsp})f(\bsr,\bsp,t)|^2.
\]
Since $f_0$ is bounded and positive, we must have $(\bsp +
\grad{\bsp}) f(\bsr,\bsp,t)\equiv0$.  Writing $f(\bsr,\bsp,t)=Z^{-1}
\exp(-|\bsp|^2/2) \phi(\bsr,\bsp,t)$ then gives $\grad{\bsp}
\phi(\bsr,\bsp,t)\equiv 0$ and the result follows.
\end{proof}

Determining the explicit solvability condition finally requires the explicit
calculation of the right hand side of (\ref{eps-1a}):
\begin{lemma}\label{L:L1N1}
For $\cL_1$ and $\cN_1$ as in (\ref{L1}) and (\ref{N1}),
and $f_0$ as in Corollary \ref{C:f0}, we have
\[
    \cL_1 f_0 = -Z^{-1}\e{-\frac{|\bsp|^2}{2}} \bsp \cdot
\big[\grad{\bsr} + \grad{\bsr}V_1(\bsr,t) \big] \rho_0(\bsr,t)
\]
and
\begin{align*}
    \cN_1(f_0,f_0)(\bsr,\bsp,t) &= -
Z^{-1} \e{-\frac{|\bsp|^2}{2}} \rho_0(\bsr,t) \bsp
\cdot \int \dd \bsr' \rho_0(\bsr',t) g(\bsr,\bsr') \grad{\bsr}V_2(\bsr,\bsr')
\end{align*}
\end{lemma}
\begin{proof}
This is a simple calculation:  For $\cL_1$ we have
\begin{align*}
    \cL_1 f_0(\bsr,\bsp,t) &= \big[-\bsp \cdot \grad{\bsr} +
\grad{\bsr}
V_1(\bsr,t) \cdot \grad{\bsp}\big] \Big[Z^{-1}\e{-\frac{|\bsp|^2}{2}}
\rho_0(\bsr,t) \Big] \\
&= Z^{-1}\e{-\frac{|\bsp|^2}{2}} \big[-\bsp \cdot \grad{\bsr}
\rho_0(\bsr,t)
\big] + \grad{\bsr} V_1(\bsr,t)\cdot \big[-\bsp \big]
Z^{-1}\e{-\frac{|\bsp|^2}{2}}\rho_0(\bsr,t),
\end{align*}
and the result follows.

The result for $\cN_1$ follows from the two identities $\int \dd
\bsp f_0(\bsr,\bsp,t) = \rho_0(\bsr,t)$ and $\grad{\bsp} f_0(\bsr,\bsp,t) =
-\bsp \e{-|\bsp|^2/2} \rho_0(\bsr,t)$.
\end{proof}

Recall that we are trying to solve (\ref{eps-1a}), and have shown
that in $L^2_{f_0^{-1}}$, $\tilde \cL$ is self adjoint with compact
resolvent, and has null space elements $\exp(-|\bsp|^2/2)
\phi(\bsr,t)$. In order for (\ref{eps-1a}) to be soluble, we
therefore require, by the Fredholm alternative, that its inner
product (in $L^2_{f_0^{-1}}$) with any element of the null space of
$\tilde \cL^*$ is zero.

Note that $\langle \e{-|\bsp|^2/2}\phi(\bsr,t),f \rangle_{f_0^{-1}} =
\langle \rho_0(\bsr,t) \phi(\bsr,t),f \rangle$ and we therefore require that the
integral with respect to $\bsp$ of the right hand side of (\ref{eps-1a}) is
zero. This is an easy corollary of Lemma \ref{L:L1N1} since the $\bsp$
dependence of both terms is of the form $\exp(-|\bsp|^2/2) \bsp$, which
integrates to zero.
\begin{corollary}
    $-\cL_1 f_0-\cN_1(f_0,f_0)$ is orthogonal to the null space of $\tilde
\cL^*$ and thus (\ref{eps-1a}) always has a solution.
\end{corollary}

Since we now know that (\ref{eps-1a}) is soluble, we can invert its
left hand side.  The standard approach would be to expand in a basis
of the eigenfunctions of $\tilde \cL$.  However, since $\bsZ_1$ and
$\bsZ_2$ (which enter $\tilde \cL$ through $\cN_0$) are unknown, we
expand in a basis of products of generalized Hermite polynomials
multiplied by a Maxwellian, which are eigenfunctions for the case
$\bsZ_1=\bsZ_2=0$.  This turns out to be sufficient as we do not
need to explicitly invert $\tilde \cL$.
\begin{definition}
We define the basis of $L^2(\R^3,\e{-|\bsp|^2/2})$, for
$\bsp=(p_1,p_2,p_3)^{\rm T}$,
\[
    P_{n,\bsa}(\bsp):=H_{a_1}(p_1)H_{a_2}(p_2) H_{a_3}(p_3),
\]
where $n \in \N$, $\bsa=(a_1,a_2,a_3)^{\rm T}$, $a_i \in \N$, $|\bsa|:=\sum a_i
= n$,
and $H_n$ are the standard one-dimensional Hermite polynomials.
\end{definition}

Since $\R^3$ is a product space and the Hermite polynomials form an orthogonal
basis of $L^2(\R,\e{-p^2/2})$, it is clear that the
$P_{n,\bsa}$ form an orthogonal basis of $L^2(\R^3, \e{-|\bsp|^2/2})$ (see
e.g.\ \cite{Dunkl01}).  We now show that $\tilde \cL$ preserves $n$, the degree
of the Hermite polynomial, when applied to $\e{-|\bsp|^2/2}P_{n,\bsa}(\bsp)$.
Hence, with a slight abuse of notation, we find
\[
\tilde \cL^{-1}
[\e{-|\bsp|^2/2}P_{n,\bsa}(\bsp)] \in \mbox{Span}\big\{
\{\e{-|\bsp|^2/2}P_{n,\bsb}(\bsp) \, | \,|\bsb|=n\} \cup \e{-|\bsp|^2/2}
\big\},
\]
where the span runs over coefficients in $\bsr$ and $t$, and we note that
$\e{-|\bsp|^2/2}$ is the $\bsp$-dependent part of the kernel of $\cL$.
\begin{lemma} \label{L:preservesN}
Denote the expansions of $\tilde{\cL} f$ and $f$ by
\begin{align*}
\tilde{\cL} f &=:
\e{-|\bsp|^2/2}\sum_{n=0}^\infty \sum_{|\bsa|=n}
\tilde\gamma_{n,\bsa}(\bsr,t) P_{n,\bsa}(\bsp)\\
f &=:
\e{-|\bsp|^2/2}\sum_{n=0}^\infty \sum_{|\bsa|=n} \gamma_{n,\bsa}(\bsr,t)
P_{n,\bsa}(\bsp).
\end{align*}
Then, for each fixed $n \geq 1$,
$\gamma_{n,\bsa}(\bsr,t) \equiv 0$ for all $\bsa$ if and only if $\tilde
\gamma_{n,\bsa}(\bsr,t) \equiv 0$ for all $\bsa$.
\end{lemma}
\begin{proof} See Appendix \ref{A:preservesN} \end{proof}

We are now in the position to determine the explicit form of $f_1$ and thus
solve (\ref{eps-1}).
\begin{lemma}\label{L:f1explicit}
   $f_1(\bsr,\bsp,t) = [\bsa(\bsr,t) \cdot \bsp + \psi(\bsr,t)] Z^{-1}
\exp(-|\bsp|^2/2)$, where $\bsa(\bsr,t)$ is given by the solution of
\begin{align*}
     &\bsa(\bsr,t) + \int \dd \bsr' g(\bsr,\bsr') \rho_0(\bsr',t)
\bsZ_1(\bsr,\bsr') \times \bsa(\bsr,t) + \rho_0(\bsr,t) \int \dd \bsr'
g(\bsr,\bsr')
\bsZ_2(\bsr,\bsr') \bsa(\bsr',t) \\
&\qquad \qquad  = -\Big[ \grad{\bsr} +
\grad{\bsr}V_1(\bsr,t) + \int \dd \bsr' \rho_0(\bsr',t) g(\bsr,\bsr')
\grad{\bsr}V_2(\bsr,\bsr')\Big] \rho_0(\bsr,t).
\end{align*}
\end{lemma}
\begin{proof}
By (\ref{eps-1a}) and Lemma \ref{L:L1N1} we have
\begin{align}
    \tilde \cL f_1 &= -\cL f_0 - \cN_1(f_0,f_0) \notag \\
    &=\Big[ \grad{\bsr} +
\grad{\bsr}V_1(\bsr,t) + \int \dd \bsr' \rho_0(\bsr',t) g(\bsr,\bsr')
\grad{\bsr}V_2(\bsr,\bsr')\Big] \rho_0(\bsr,t) \cdot \bsp Z^{-1}
\e{-|\bsp|^2/2}
\notag \\
&=: \tilde \bsa(\bsr,t) \cdot \bsp Z^{-1}  \e{-|\bsp|^2/2}. \label{atildeDefn}
\end{align}
Hence, by Lemma \ref{L:preservesN} and the definitions $P_{1,\bse_j}=p_j$ for
$\bse_j$ the standard unit vectors, it follows that
$f_1(\bsr,\bsp,t) = [\bsa(\bsr,t) \cdot \bsp + \psi(\bsr,t)] Z^{-1}
\exp(-|\bsp|^2/2)$ for some $\bsa$.  Evaluating each of the terms of $\tilde \cL
f_1$ then gives, firstly from (\ref{L0}),
\begin{align*}
    \cL_0 f_1 &= \div{\bsp}\big[(\bsp + \grad{\bsp}) [\bsa(\bsr,t)
\cdot \bsp + \psi(\bsr,t)] Z^{-1} \e{-|\bsp|^2/2} \big]\\
&= \div{\bsp}\big[Z^{-1} \e{-|\bsp|^2/2} \grad{\bsp}
[\bsa(\bsr,t)
\cdot \bsp + \psi(\bsr,t)]  \big]\\
&= \div{\bsp}\big[Z^{-1} \e{-|\bsp|^2/2} \bsa(\bsr,t)] = - Z^{-1}
\e{-|\bsp|^2/2} \bsp \cdot \bsa(\bsr,t).
\end{align*}

Using the explicit form of $\cN_0(f_0,f_1)$, as given by Lemma \ref{L:N0}, gives
\begin{align*}
    \cN_0(f_0,f_1) &= \div{\bsp} \Big[ \int \dd \bsr' g(\bsr,\bsr')
\rho_0(\bsr',t) \bsZ_1(\bsr,\bsr')\\
&\qquad \qquad \qquad \qquad \times (\bsp +
\grad{\bsp})\big[ [\bsa(\bsr,t)
\cdot \bsp + \psi(\bsr,t)] Z^{-1} \e{-|\bsp|^2/2}\big] \Big]\\
&= \div{\bsp}\Big[ \int \dd \bsr' g(\bsr,\bsr') \rho_0(\bsr',t)
\bsZ_1(\bsr,\bsr') \times \bsa(\bsr,t) Z^{-1}
\e{-|\bsp|^2/2} \Big].
\end{align*}
For a general matrix $\bsZ(\bsr)$, we have
\begin{align*}
    &\div{\bsp} \big[ \bsZ(\bsr) \bsa(\bsr,t) Z^{-1}
\e{-|\bsp|^2/2} \big] = \sum_{i,j=1}^3 \partial_{p_i}
\big[Z_{ij}(\bsr,t) a_j Z^{-1} \e{-|\bsp|^2/2} \big] \\
& \qquad =
-\sum_{i,j=1}^3 Z_{ij}(\bsr,t) p_i a_j Z^{-1} \e{-|\bsp|^2/2}
= - \bsp \cdot \bsZ(\bsr,t) \bsa(\bsr,t)Z^{-1} \e{-|\bsp|^2/2},
\end{align*}
and hence
\[
    \cN_0(f_0,f_1) = - Z^{-1} \e{-|\bsp|^2/2} \bsp \cdot \int \dd \bsr'
g(\bsr,\bsr') \rho_0(\bsr',t) \bsZ_1(\bsr,\bsr') \bsa(\bsr,t).
\]

For the third term, we take $\cN(f_1,f_0)$ as given by Lemma \ref{L:N0} and
note that $(\bsp' + \grad{\bsp'})f_1(\bsr',\bsp',t) = \bsa(\bsr',t)
\e{-|\bsp'|^2/2}$, giving
\begin{align*}
    \cN_0(f_1,f_0) &= - Z^{-1} \e{-|\bsp|^2/2} \rho_0(\bsr,t) \bsp
\cdot
\int \dd \bsr' \dd \bsp' g(\bsr,\bsr') \bsZ_2(\bsr,\bsr') \bsa(\bsr',t)
Z^{-1} \e{-|\bsp|^2/2}\\
& =- Z^{-1} \e{-|\bsp|^2/2} \rho_0(\bsr,t) \bsp \cdot
\int \dd \bsr' g(\bsr,\bsr') \bsZ_2(\bsr,\bsr') \bsa(\bsr',t).
\end{align*}
Collecting the three terms gives
\begin{align*}
    \tilde \bsa &= -\bsa(\bsr,t) - \int \dd \bsr' g(\bsr,\bsr') \rho_0(\bsr',t)
\bsZ_1(\bsr,\bsr') \times \bsa(\bsr,t) \\
&\qquad - \rho_0(\bsr,t) \int \dd \bsr'
g(\bsr,\bsr') \bsZ_2(\bsr,\bsr') \bsa(\bsr',t),
\end{align*}
and the result follows by (\ref{atildeDefn}).
\end{proof}

\subsection{Solution of the $\epsilon^0$ equation}
We have, by Corollary \ref{C:f0} and Lemma \ref{L:f1explicit}, that
\begin{align*}
    f_0(\bsr,\bsp,t) &=Z^{-1} \exp(-|\bsp|^2/2) \rho_0(\bsr, t), \\
    f_1(\bsr,\bsp,t) &= [\bsa(\bsr,t) \cdot \bsp + \psi(\bsr,t)]
Z^{-1} \exp(-|\bsp|^2/2).
\end{align*}
We now show that the evolution equation for $\rho_0$
is given by the solvability condition for the equation corresponding to
$\epsilon^0$, namely (\ref{eps0}).  We begin by rewriting (\ref{eps0}) as
\be \label{eps0a}
    -\tilde \cL f_2 = \cL_1 f_1 + \cN_0(f_1,f_1) + \cN_1(f_1,f_0) +
\cN_1(f_0,f_1) - \partial_t f_0.
\ee
We then require that the right hand side is orthogonal to the null space of
$\tilde \cL^*$ in $L^2_{f_0^{-1}}$ which, by Lemma \ref{L:nullLtildestar},
is equivalent to it being orthogonal to constants
in $\bsp$ in the unweighted space $L^2$.  Hence, using the divergence theorem
and the explicit forms
of $\cN_0$ and $\cN_1$ as given by (\ref{N0}) and (\ref{N1}), the requirement
reduces to $\int \dd \bsp (\cL_1 f_1 - \partial_t f_0)=0$, and we need only
calculate these two terms.  However, for completeness and later use, we also
calculate the $\cN_0$ and $\cN_1$ terms:
\begin{lemma} \label{L:eps0terms}
    For $\cL_1$, $\cN_0$ and $\cN_1$ as in (\ref{L1})--(\ref{N1}), $f_0$ as in
Corollary \ref{C:f0} and $f_1$ as in Lemma \ref{L:f1explicit}, we have
\begin{align*}
    \cL_1 f_1 &= Z^{-1}\e{-|\bsp|^2/2} \Big[ -
    \bsp\cdot \grad{\bsr}\bsa(\bsr,t) \bsp
    - \bsp \cdot \grad{\bsr} \psi(\bsr,t)  \\
    & \qquad \qquad \qquad\qquad +
\grad{\bsr}V_1(\bsr,t) \cdot \big[\bsa(\bsr,t) - \bsp \big(\bsp
\cdot \bsa(\bsr,t) + \psi(\bsr,t) \big)\big] \Big]\\
    \cN_1(f_0,f_1) &= Z^{-1}\e{-|\bsp|^2/2} \int \dd \bsr'
\rho_0(\bsr',t) g(\bsr,\bsr') \grad{\bsr}V_2(\bsr,\bsr') \\
& \qquad\qquad\qquad\qquad \cdot \Big[
\bsa(\bsr,t) -  \bsp ( \bsp\cdot \bsa(\bsr,t) +
\psi(\bsr,t))\Big]\\
    \cN_1(f_1,f_0) &= -Z^{-1}\e{-|\bsp|^2/2} \rho_0(
\bsr,t)  \bsp \cdot\int \dd \bsr'
\psi(\bsr',t) g(\bsr,\bsr') \grad{\bsr}V_2(\bsr,\bsr') \\
    \cN_0(f_1,f_1)&=-Z^{-1} \e{-|\bsp|^2/2}\bsp \cdot \Big[ \int \dd
\bsr' g(\bsr,\bsr') \psi(\bsr',t)
\bsZ_1(\bsr,\bsr')\Big]\bsa(\bsr,t) \\
& \qquad +  \e{-|\bsp|^2/2} \Big[ \int \dd \bsr'
g(\bsr,\bsr') \bsZ_2(\bsr,\bsr') \bsa(\bsr',t)\Big]\\
& \qquad\qquad\qquad\qquad\qquad\qquad \cdot \Big[
\bsa(\bsr,t) - \bsp ( \bsp\cdot \bsa(\bsr,t) +
\psi(\bsr,t))\Big]
\end{align*}
\end{lemma}
\begin{proof}
We begin with $\cL_1 f_1$, which is given by
\begin{align*}
    \cL_1 f_1(\bsr,\bsp,t) &= [-\bsp \cdot \grad{\bsr} +
\grad{\bsr}V_1(\bsr,t)
\cdot \grad{\bsp}] \big[ Z^{-1} \e{-|\bsp|^2/2} \big( \bsa(\bsr,t)
\cdot \bsp + \psi(\bsr,t) \big) \big]
\end{align*}
Simple calculations show that
\begin{align}
\grad{\bsp} [Z^{-1}\e{-|\bsp|^2/2} (\bsp \cdot \bsa(\bsr,t) +
\psi(\bsr,t))]
&= Z^{-1}\e{-|\bsp|^2/2}  [- \bsp + \grad{\bsp}]
(\bsp \cdot \bsa(\bsr,t) + \psi(\bsr,t)) \notag\\
&\hspace*{-20mm}= Z^{-1}\e{-|\bsp|^2/2} \big[\bsa(\bsr,t) - \bsp
\big(\bsp \cdot \bsa(\bsr,t) + \psi(\bsr,t) \big)\big] \label{gradpf1} \\
-\bsp \cdot \grad{\bsr} (Z^{-1} \e{-|\bsp|^2/2}
\psi(\bsr,t)) &= -Z^{-1}
\e{-|\bsp|^2/2} \bsp \cdot \grad{\bsr}  \psi(\bsr,t) \notag
\end{align}
This gives three of the terms.  The calculation of the fourth is easier in
coordinates.  Using that $\bsp \cdot \grad{\bsr} = \sum_{i=1}^3 p_i
\partial_{r_i}$ and $\bsp \cdot \bsa = \sum_{j=1}^3 p_j a_j$ gives
\begin{align*}
    -\bsp \cdot \grad{\bsr}[ Z^{-1} \e{-|\bsp|^2/2} \bsp
\cdot
\bsa(\bsr,t)] &=
- Z^{-1} \e{-|\bsp|^2/2}  \bsp \cdot \grad{\bsr}[ \bsp
\cdot
\bsa(\bsr,t)] \\
&=-Z^{-1} \e{-|\bsp|^2/2} \sum_{i,j=1}^3 p_i \partial_{r_i}
[p_j
a_j(\bsr,t)]\\
&=-Z^{-1} \e{-|\bsp|^2/2} \bsp \cdot \grad{\bsr}
\bsa(\bsr,t) \bsp,
\end{align*}
which gives the result for $\cL_1
f_1$.

The expressions for $\cN_1(f_0,f_1)$ and $\cN_1(f_1,f_0)$ result from the
trivial identities $\int \dd \bsp'
f_0(\bsr',\bsp',t) = \rho_0(\bsr,t)$, $\int \dd \bsp'
f_1(\bsr',\bsp',t) = \psi(\bsr,t)$, $\grad{\bsp} f_0(\bsr,\bsp,t) =
-\bsp f_0(\bsr,\bsp,t)$ and (\ref{gradpf1}).

Finally, for $\cN_0(f_1,f_1)$ we use the trivial identity $\int \dd
\bsp' f_1(\bsr',\bsp',t) = \psi(\bsr,t)$, and that $(\bsp +
\grad{\bsp})f_1(\bsr,\bsp,t)=
Z^{-1} \e{-|\bsp|^2/2} \grad{\bsp}[ \bsp \cdot \bsa(\bsr,t) +
\psi(\bsr,t)] = Z^{-1} \e{-|\bsp|^2/2} \bsa(\bsr,t)$, and so  $ \int
\dd \bsp' (\bsp' + \grad{\bsp'})f_1(\bsr,\bsp',t) = \bsa(\bsr',t)$.
Then
\begin{align*}
    \cN_0(f_1,f_1)&=\div{\bsp} \Big[ \int \dd \bsr' g(\bsr,\bsr')
\psi(\bsr',t)
\bsZ_1(\bsr,\bsr') \times \bsa(\bsr,t) Z^{-1} \e{-|\bsp|^2/2} \Big] \\
& + \div{\bsp} \Big[ \int \dd \bsr' g(\bsr,\bsr')
\bsZ_2(\bsr,\bsr') \bsa(\bsr',t) \times
Z^{-1} \e{-|\bsp|^2/2}\big(\bsp\cdot \bsa(\bsr,t) + \psi(\bsr,t)\big)
\Big],
\end{align*}
both terms of which are of the form $\div{\bsp} \bsv(\bsr,t)
\phi(\bsr,\bsp,t)$, where $\bsv$ is a vector.  Using $\div{\bsp}
\bsv(\bsr,t) \phi(\bsr,\bsp,t)= \bsv(\bsr,t) \cdot
\grad{\bsp}\phi(\bsr,\bsp,t)$, $\grad{\bsp}
\e{-|\bsp|^2/2} = - \bsp
\e{-|\bsp|^2/2}$ and (\ref{gradpf1}) completes the proof.
\end{proof}

Recall that we wish to solve (\ref{eps0a}), and require that $\int \dd \bsp
(\cL_1 f_1 - \partial_t f_0)=0$.
The identities $\int \dd \bsp Z^{-1} \e{-|\bsp|^2/2}=1$, $\int \dd
\bsp Z^{-1} \e{-|\bsp|^2/2} p_i p_j = \delta_{ij}$ show that
\begin{align*}
    \int \dd \bsp Z^{-1} \e{-|\bsp|^2/2} [\bsa(\bsr,t) -
\bsp (\bsp\cdot \bsa)] &= \bsa(\bsr,t) - \int
\dd \bsp Z^{-1} \e{-|\bsp|^2/2} \bsp \sum_{i=1}^3 p_j a_j(\bsr,t)\\
&= \bsa(\bsr,t)-\bsa(\bsr,t)=0
\end{align*}
(which also follows from these terms resulting from the $\grad{\bsr} V_1\cdot
\grad{\bsp}$ term in $\cL_1$ and the divergence theorem) and
\[
    \int \dd \bsp Z^{-1} \e{-|\bsp|^2/2} \bsp \cdot \grad{\bsr}
\bsa(\bsr,t) \bsp = \int \dd \bsp Z^{-1} \e{-|\bsp|^2/2} \sum_{i,j=1}^3
p_i p_j \partial_{r_j} a_j(\bsr,t) = \div{\bsr} \bsa(\bsr,t),
\]
and hence $\int \dd \bsp \cL_1 f_1 = - \div{\bsr} \bsa(\bsr,t)$.
Since $\int \dd \bsp \partial_t
f_0(\bsr,\bsp,t) = \partial_t \rho_0(\bsr,t)$, the solvability condition becomes
\be \label{rhoEvol}
     \partial_t \rho_0(\bsr,t)= -\div{\bsr}\bsa(\bsr,t),
\ee which is precisely the equation describing the one-body position
distribution evolution for the Smoluchowski equation, as given in
Theorem \ref{T:MainTheorem}.

\subsection{Solution of the $\epsilon^1$ equation}  \label{S:epsilon1}
We now demonstrate that $\psi(\bsr,t) \equiv 0$ if
$\psi(\bsr,0)\equiv 0$. This should result from the solvability
condition for the $\epsilon^1$ equation, which has the form
\[
 - \tilde \cL f_3 = \cL_1 f_2 + \cN_0(f_2,f_1) + \cN_0(f_1,f_2) +
\cN_1(f_2,f_0) + \cN_1(f_0,f_2) + \cN_1(f_1,f_1) - \partial_t f_1.
\]
Since once again the $\cN_0$ and $\cN_1$ terms do not contribute to
the Fredholm alternative calculation, we have
\[
    \int \dd \bsp (\cL_1 f_2 - \partial_t f_1) =0.
\]
From (\ref{eps0a}), we have that
\be \label{f2L-1}
    f_2 = (-\tilde \cL)^{-1} [ \cL_1 f_1 + \cN_0(f_1,f_1) + \cN_1(f_1,f_0) +
    \cN_1(f_0,f_1) - \partial_t f_0],
\ee
with $\partial_t f_1 = Z^{-1} \e{-|\bsp|^2/2}[\bsp \cdot\partial_t
\bsa(\bsr,t) + \partial_t\psi(\bsr,t)]$ and the remaining terms given by Lemma
\ref{L:eps0terms}.

For the $\cL_1 f_2$ term, we have $\cL_1 f_2 = [-\bsp \cdot
\grad{\bsr} + \grad{\bsr}V_1(\bsr,t) \cdot \grad{\bsp}] f_2(\bsr,\bsp,t)$, and,
by
the divergence theorem, the second term vanishes upon integration.  Hence we are
interested only in
\[
    \int \dd \bsp \big[ Z^{-1} \e{-|\bsp|^2/2} \partial_t
\psi(\bsr,t)+ \bsp \cdot \grad{\bsr} f_2(\bsr,\bsp,t)  \big] =
\partial_t \psi(\bsr,t) + \int \dd \bsp \, \bsp \cdot \grad{\bsr}
f_2(\bsr,\bsp,t).
\]
Since $H_{1,\bse_j}(p) = p_j$, the only terms from $f_2$ which contribute
to the integral are of the form $\bsp \cdot \bsa_2(\bsr,t)
Z^{-1} \e{-|\bsp|^2/2}$, i.e $\bsp \cdot \cP_1(f_2)$ where $\cP_1$ is the
projection onto $\bsp$, i.e.\ $\cP_1 f= \int \dd \bsp \, \bsp
f(\bsr,\bsp,t)$. By (\ref{f2L-1}) and Lemma
\ref{L:preservesN}, it therefore suffices to consider only terms of the form
$-\bsp \cdot \tilde \bsa_2(\bsr,t) Z^{-1} \e{-|\bsp|^2/2}$ in
$ \cL_1 f_1 + \cN_0(f_1,f_1) + \cN_1(f_1,f_0) + \cN_1(f_0,f_1) - \partial_t
f_0$, i.e.\ $\tilde \bsa_2(\bsr,t)= -\cP_1 \big(\cL_1 f_1 + \cN_0(f_1,f_1) +
\cN_1(f_1,f_0) + \cN_1(f_0,f_1) - \partial_t f_0 \big)$.

By Lemma \ref{L:eps0terms} and $\partial_t f_0(\bsr,\bsp,t)=Z^{-1}
\e{-|\bsp|^2/2}\partial_t\rho_0(\bsr,t)$, we have
\[
    \cP_1(-\tilde \cL f_2) = \cP_1\big( \cL_1 f_1 + \cN_0(f_1,f_1) +
\cN_1(f_1,f_0) + \cN_1(f_0,f_1) - \partial_t f_0 \big) =:
-\tilde\bsa_2(\bsr,t),
\]
with
\begin{align*}
    \tilde \bsa_2(\bsr,t):&=\Big[ \grad{\bsr} \psi(\bsr,t)
+
\psi(\bsr,t) \grad{\bsr} V_1(\bsr,t) + \int
\dd \bsr' \rho_0(\bsr,t) g(\bsr,\bsr') \grad{\bsr}  V_2(\bsr,\bsr')
\psi(\bsr,t)\\
& \qquad  \qquad + \rho_0(\bsr,t) \Big( \int \dd \bsr'
\psi(\bsr',t)
g(\bsr,\bsr') \grad{\bsr} V_2(\bsr,\bsr') \Big)\\
& \qquad  \qquad  + \Big( \int \dd \bsr'
\psi(\bsr',t) g(\bsr,\bsr') \bsZ_1(\bsr,\bsr') \Big) \bsa(\bsr,t) \\
& \qquad \qquad +
\Big( \int \dd \bsr'  g(\bsr,\bsr') \bsZ_2(\bsr,\bsr')\bsa(\bsr',t) \Big)
\psi(\bsr,t) \Big],
\end{align*}
where $\bsa(\bsr,t)$ is given by Lemma \ref{L:f1explicit}.
Thus, by (the proof of) Lemma \ref{L:f1explicit},
\[
   \cP_1 f_2(\bsr,\bsp,t) =\bsa_2(\bsr,t)
\]
with $\bsa_2$ the solution of
\begin{align*}
     -\tilde \bsa_2(\bsr,t) &= \bsa_2(\bsr,t) + \int \dd \bsr' g(\bsr,\bsr')
\rho_0(\bsr',t)
\bsZ_1(\bsr,\bsr') \times \bsa_2(\bsr,t) \\
& \qquad + \rho_0(\bsr,t)\int \dd \bsr'
g(\bsr,\bsr')
\bsZ_2(\bsr,\bsr') \bsa_2(\bsr',t)
\end{align*}
Hence, for the $\epsilon^1$ equation to be solvable,
\begin{align*}
    0 &=  \partial_t \psi(\bsr,t) + \int \dd \bsp \, \bsp \cdot
\grad{\bsr} \cP_1f_2(\bsr,\bsp,t)\\
& = \partial_t \psi(\bsr,t) +
\int \dd \bsp Z^{-1} \e{-|\bsp|^2/2}  \bsp \cdot
\grad{\bsr}[\bsp \cdot \bsa_2(\bsr,t)]= \partial_t \psi(\bsr,t) + \div{\bsr}
\bsa_2(\bsr,t)
\end{align*}
or
\be
    \partial_t \psi(\bsr,t)=-\div{\bsr} \bsa_2(\bsr,t).
    \label{psiEvol}
\ee

To ensure that $\psi(\bsr,t)\equiv0$, we first note that
$\psi(\bsr,0)\equiv 0$ is equivalent to assuming that the initial
condition $f^{(1)}(\bsr,\bsp,0)$ is independent of $\epsilon$.  For
this to hold for all $t$, it is necessary to show that
(\ref{psiEvol}) is dissipative (or, since $\partial_t \int \dd \bsr
\psi(\bsr,t)=0$, that (\ref{psiEvol}) is non-negativity preserving).

The proof in the linear ($V_2$, $\bsZ_1$, $\bsZ_2$ all zero) case is
trivial, as it turns out that $\psi$ and $\rho$ satisfy the same
equation.  Thus, since the Smoluchowski equation for $\rho$ must be
non-negativity preserving, so must the equation for $\psi$. The
proof in the general case is complicated both by the equations for
$\rho$ and $\psi$ not being identical (due to the non-linear terms)
and by needing to prove dissipativity results for the resulting
non-linear operators.  In general (for $\bsZ_2 \neq 0$), the
equations are not even explicit as one needs to solve the Fredholm
integral equations for $\bsa$ and $\bsa_2$. However, since the full
friction tensor is positive-definite, one would expect
(\ref{psiEvol}) to be a parabolic PDE and so, for potentials and
hydrodynamic interaction terms with sufficient bounded derivatives,
the result should follow from standard PDE theory, see e.g.\
\cite{Walter86}. We therefore assume that $V_1$, $V_2$, $\bsZ_1$ and
$\bsZ_2$ are such that that if $\psi(\bsr,0)\equiv 0$ then
$\psi(\bsr,t)\equiv 0$ for all $t\geq 0$.

\section{The Smoluchowski equation} \label{S:Smoluchowski}
We are now in a position to state our rigorously derived
Smoluchowski equation. For ease of comparison with existing results,
we return to the original scalings of time and potentials.
\begin{theorem}[Smoluchowski equation] \label{T:MainTheorem}
Suppose $f^{(1)}(\bsr,\bsp,0)=f_0(\bsr,\bsp,0)=Z^{-1} \e{-|\bsp|^2/
2}\rho_0(\bsr,0)$ is independent of $\epsilon$.  Suppose further
that $\rho_0(\bsr,0)$, $U_j$ and $\bsZ_j$, $j=1,2$ are such that the solutions
of (\ref{dtf1abstract}), (\ref{rhoEvol}) and (\ref{psiEvol}) exist for times
$[0,t_0]$ and that (\ref{psiEvol}) is non-negativity preserving. Then,
up to errors of $\mathcal{O}(\epsilon^2)$, the dynamics of the one-body
position distribution are given (in the original timescale) for $\tau \in [0,
m\gamma/(k_BT) t_0]$ by
\[
     \partial_\tau \rho(\bsr,\tau)= -\tfrac{k_BT}{m\gamma}
\div{\bsr}\bsa(\bsr,\tau),
\]
where $\bsa(\bsr,\tau)$ is the solution to
\begin{align}
     & \bsa(\bsr,\tau) + \int \dd \bsr' g(\bsr,\bsr') \rho(\bsr',\tau)
\bsZ_1(\bsr,\bsr')\times \bsa(\bsr,\tau) + \rho(\bsr,\tau) \int \dd \bsr'
g(\bsr,\bsr')
\bsZ_2(\bsr,\bsr') \bsa(\bsr',\tau) \notag \\
&\qquad \qquad  = -\Big[ \grad{\bsr} + \tfrac{1}{k_BT} \Big(
\grad{\bsr}U_1(\bsr,\tau) + \int \dd \bsr' \rho(\bsr',\tau) g(\bsr,\bsr')
\grad{\bsr}U_2(\bsr,\bsr') \Big) \Big] \rho(\bsr,\tau). \label{aFredholm}
\end{align}
\end{theorem}
\begin{proof}
The evolution equation is given by (\ref{rhoEvol}), which is the
solvability condition for (\ref{eps0}), and $\bsa(\bsr,t)$ is given
by Lemma \ref{L:f1explicit}.  Returning to the original timescale
introduces the factor of $\mu=k_BT/(m\gamma)$ in the right hand
side. We also replace $V$ by its original value of $U/(k_BT)$ where
$\bsX_i=-\grad{\bsr_i}U(\bsr^N)$. The conditions on (\ref{psiEvol})
and the initial condition ensure that, using the notation of Lemma
\ref{L:f1explicit}, $\psi(\bsr,t)\equiv0$ for all times.  Hence
$\rho(\bsr,\tau)=\rho_0(\bsr,\tau) + \cO(\epsilon^2)$.
\end{proof}

We note here that the assumptions on the initial condition and on
the existence of solutions are analogous to those made for the
Boltzmann equation, see e.g.\ \cite{EspositoLebowitzMarra99}.  We
expect that proving such assumptions hold for a physically
interesting range of potentials and friction tensors would be a
formidable problem in its own right, and is beyond the scope of the
present study.  An analysis of the corresponding problem for the
Boltzmann equation is given in~\cite{DiPernaLions89}.

To demonstrate the connection to existing formulations, we assume
that $\bsZ_2 \equiv0$, which allows us to find $\bsa$ explicitly. We
then have:
\begin{corollary} \label{C:Z2Zero}
    Under the same assumptions as in Theorem \ref{T:MainTheorem}, if
$\bsZ_2\equiv 0$, the one-body position dynamics are, up to errors
of $\mathcal{O}(\epsilon^2)$, governed by
\begin{align}
 \partial_\tau \rho(\bsr,\tau) &= \div{\bsr}\Big(\bsD(\bsr,\tau)
\Big[\grad{\bsr} \rho(\bsr,\tau) + \tfrac{1}{k_BT}
\rho(\bsr,\tau) \grad{\bsr}V_1(\bsr,\tau)
\notag \\
&\qquad \qquad \qquad \qquad \qquad + \tfrac{1}{k_BT} \int \dd \bsr'
\rho^{(2)}(\bsr, \bsr',\tau) \grad{\bsr}V_2(\bsr,\bsr')\Big] \Big),
\label{Smoluchowski}
\end{align}
where we have defined
$\rho^{(2)}(\bsr,\bsr',\tau):=\rho(\bsr,\tau)\rho(\bsr',\tau)g(\bsr,\bsr',\tau)$,
as it would be for the Enskog approximation, and the $3 \times 3$
diffusion tensor $\bsD$ is given by
\[
    \bsD(\bsr,\tau) = \frac{k_BT}{m\gamma} \Big[\bsone + \int \dd \bsr'
g(\bsr,\bsr') \rho(\bsr',\tau)
\bsZ_1(\bsr,\bsr') \Big]^{-1}
\]
\end{corollary}

Corollary \ref{C:Z2Zero} gives a one-body Smoluchowski equation with a
novel form for the diffusion tensor $\bsD$.  As is clear from the notation,
$\bsD(\bsr,\tau)$ not only depends on the position but also on the time.  This
time-dependence is present through the time-dependence of $\rho$, against which
the two-body terms must be averaged.

One obvious question is whether $\bsD$ is positive definite.  A simple
calculation shows that this is indeed the case.  Note that $\bsone + \sum_{j\neq
1} \bsZ_1(\bsr_1,\bsr_j)$ is positive definite (since it is a principal minor of
$\bGa$, which is positive definite).  Hence, for any $\bsv(\bsr,\tau)$, we
have, for some $\delta>0$,
\begin{align*}
& \bsv(\bsr_1,\tau) \cdot [\bsone + \sum_{j\neq 1}
\bsZ_1(\bsr_1,\bsr_j)] \bsv(\bsr_1,\tau) \geq \delta |\bsv(\bsr_1,\tau)|^2 \\
& \Rightarrow \bsv(\bsr_1,\tau) \cdot \int \dd \bsr_2
\rho(\bsr_2,\tau)g(\bsr_1,\bsr_2) [\tfrac{1}{N-1} \bsone +
\bsZ_1(\bsr_1,\bsr_2)] \bsv(\bsr_1,\tau) \geq \delta |\bsv(\bsr_1,\tau)|^2,
\end{align*}
where the proof is virtually identical to that of Lemma \ref{L:IntegralPD},
except we do not integrate over $\bsr_1$.  Since $\int \dd \bsr_2
\rho(\bsr_2,\tau) g(\bsr,\bsr_2)= N-1$, and $k_BT$, $m$ and $\gamma$ are
positive, this is equivalent to $\bsD^{-1}$, and hence $\bsD$, being positive
definite.

We now compare our result with that derived by Rex and L\"owen
\cite[(5)--(8)]{RexLowen09}. As demonstrated in Figure
\ref{Fig:NonCommutative}, their Smoluchowski equation is derived
from the $N$-body Smoluchowski equation for pairwise additivity of
both the potential (our Assumption 1) and diffusion tensor.  The
second assumption is analogous to our Assumption 2, but not
equivalent, as the inverse of a matrix (recall $\bGa \bsD= k_BT/m
\bsone$) with pairwise terms need not contain only pairwise terms.
However, there are situations where the two assumptions are essentially
equivalent, such as in a diffuse colloid system.
The
underlying assumption then is that there exists an additional small
parameter, say $\lambda$, with $1 \gg \lambda \gg \epsilon$ and such
that $\bsZ_1 = \cO(\lambda)$. Then, up to errors of
$\cO(\lambda^2)$, $\bsD(\bsr,\tau)=\tfrac{k_BT}{m\gamma}[ \bsone
-\int \dd \bsr' g(\bsr,\bsr') \rho(\bsr',\tau) \bsZ_1(\bsr,\bsr')]$.
and thus the diffusion tensor is a two-body one. We note that the
analogue of Assumption 3 is
$\rho^{(2)}(\bsr_1,\bsr_2,t)=\rho(\bsr_1,t)\rho(\bsr_2,t)
g(\bsr_1,\bsr_2)$, which can be seen by integrating out the momentum
dependence.

The simplest case is that in which both $\bGa:=\gamma \bsone$ and $\bsD:=D_0
\bsone$ are proportional to the identity matrix, when we have the standard
definition $D_0=k_BT/(m\gamma)$.  In this case it is easy to check that the two
formulations agree (see also \cite{Archer09}).  This is unsurprising as both
the difficulties and interest in this analysis lie with the non-uniform terms
in the friction tensor.

To demonstrate that the two formulations differ in general, we consider the
simple example used in Corollary \ref{C:Z2Zero}.  In addition, we assume
the existence of a parameter $\lambda$, as described above.  Then, by the block
diagonal form of $\bGa$, $\bsD$ is also block diagonal with blocks $D_0 (1 -
\sum_{\ell \neq i} \bsZ_1(\bsr_i,\bsr_\ell))$, i.e.\ in the notation of
\cite{RexLowen09} $\bsw_{11}=-\bsZ_1$.  The result to compare with
(\ref{Smoluchowski}) is (see \cite{RexLowen09})
\begin{align*}
    \partial_\tau \rho(\bsr,\tau) &= D_0 \div{\bsr} \Big[
    \grad{\bsr} \rho(\bsr,\tau) + \tfrac{1}{k_BT}
\rho(\bsr,\tau) \grad{\bsr}V_1(\bsr,\tau)  \\
& \quad + \tfrac{1}{k_BT} \int \dd \bsr'
\rho^{(2)}(\bsr, \bsr',\tau) \grad{\bsr}V_2(\bsr,\bsr')\\
&\quad - \int \dd \bsr' \bsZ_1(\bsr,\bsr') \Big( \grad{\bsr}
\rho^{(2)}(\bsr,\bsr',\tau) + \tfrac{1}{k_BT} \grad{\bsr}[V_1(\bsr,\tau) +
V_2(\bsr,\bsr')] \rho^{(2)}(\bsr,\bsr',\tau) \\
& \qquad \qquad \qquad \qquad \qquad + \tfrac{1}{k_BT} \int \dd \bsr''
\rho^{(3)}(\bsr,\bsr',\bsr'',\tau) \grad{\bsr}V_2(\bsr,\bsr'') \Big) \Big].
\end{align*}
Using the approximate two-body form of $\bsD$ in (\ref{Smoluchowski}) gives
\begin{align*}
\partial_\tau \rho(\bsr,\tau) &= D_0 \div{\bsr}
\Big[\grad{\bsr}
\rho(\bsr,\tau) + \tfrac{1}{k_BT}
\rho(\bsr,\tau) \grad{\bsr}V_1(\bsr,\tau) \\
&\quad + \tfrac{1}{k_BT} \int \dd \bsr'
\rho^{(2)}(\bsr, \bsr',\tau) \grad{\bsr}V_2(\bsr,\bsr')
\\
& \quad -\int \dd \bsr' \bsZ_1(\bsr,\bsr') \Big( \rho(\bsr',\tau)
g(\bsr,\bsr') \grad{\bsr}\rho(\bsr,\tau) + \tfrac{1}{k_BT}
\rho^{(2)}(\bsr,\bsr',\tau) \grad{\bsr} V_1(\bsr,\tau)\\
& \qquad \qquad \qquad \qquad \quad + \tfrac{1}{k_BT}
\rho(\bsr',t)g(\bsr,\bsr') \int \dd \bsr'' \rho^{(2)}(\bsr,\bsr'',\tau)
\grad{\bsr}V_2(\bsr,\bsr'') \Big) \Big],
\end{align*}
and it is clear that the two formulations are not, in general,
equivalent.
See Figure \ref{Fig:NonCommutative} for a diagrammatic
representation of the difference in the formalisms.

Interestingly, despite their obvious differences, in the overdamped
limit both formulations are accurate to $\cO(\epsilon^2)$. There
does not seem to be any mathematical or physical justification to
say that one of them is `more correct' than the other. However,
these differences make it clear that the two processes, (i)
adiabatically eliminating the fast momentum variable and (ii)
integrating over all but one particle's coordinates, do not commute.
It is worth noting that the need for knowledge of $\rho^{(3)}$ in
the first case stems from the explicit coupling of the two-body
diffusion tensor and potential in the $N$-body Smoluchowski
equation. In contrast, the potential and friction tensor are not
explicitly coupled in the $N$-body Kramers equation, and thus only
$\rho^{(2)}$ is required.  This is a partial explanation of why the
resulting equations must be different. 

It would be interesting to
perform numerical studies to see if one can quantify the
differences, i.e.\ if one can determine the magnitude of the
difference in the $\cO(\epsilon^2)$ terms.
The first form above has been implemented numerically as a DDFT by making
the further approximation that the term involving the
many-body potential is given by its value in an equilibrium system
with the same one-body density \cite{RexLowen09}. This introduces additional,
uncontrolled errors and as such a direct comparison with the new formulation
presented here, which requires no further approximations, would likely be
uninformative. For further numerical studies, including comparison with the full
underlying stochastic dynamics, demonstrations of the large qualitative and
quantitative effects of hydrodynamic interactions, and a novel DDFT
including inertial effects, see \cite{GNSPK12}.

We close by stating a result which is most useful when reducing from
a phase-space dynamical density functional theory to one in only
position space:
\begin{corollary} \label{C:epsilon^2}
    Terms proportional to Hermite polynomials of order 2 and higher in
$\bsp$ enter $f(\bsr,\bsp,t)$ at most with order $\epsilon^2$.
\end{corollary}

\section{Conclusions and open problems} \label{S:Conclusions}
Our main result is that, for suitable two-body potentials and
friction tensors, and using the Enskog approximation in the limit of
small $\epsilon$, the leading-order solution to (\ref{dtf1abstract})
is given by Theorem \ref{T:MainTheorem}.  This is a novel
Smoluchowski-type equation with a new definition of the one-body
diffusion tensor.  In addition, the Hilbert expansion studied in
Section \ref{S:Hilbert} allows us to show rigorously that a term
typically neglected by heuristic arguments in the derivation of DDFT
\cite{Archer06,Archer09} is indeed negligible in the overdamped
limit; see Corollary \ref{C:epsilon^2}.  However, these results have
only been shown to hold when the initial condition is independent of
$\epsilon$, along with assuming that $g$ is independent of
$\epsilon$ and $\bsp$.  We now discuss how removing these
assumptions should be tackled.


The assumption that the initial condition is independent of
$\epsilon$ was made for convenience as it allows for analytical
progress. In general, however, the question of how the initial
condition for Kramers equation is related to the correct
corresponding initial condition for the Smoluchowski equation needs
to be addressed. As with the Boltzmann equation (cf.
\cite{EspositoLebowitzMarra99}) and also noted in the discussion
following (\ref{fexpansion}), we would expect a boundary layer in
time (much shorter than the macroscopic timescale discussed in this
work), over which the given initial condition is attracted to one
with Gaussian momentum distribution. However, this introduces
additional complications, that can be studied by modifying
appropriately the Hilbert expansion, introducing terms that account
for the boundary layer and decay exponentially in time, see e.g.\
\cite{Banasiak06}. We note that even if the non-negativity
preserving assumption of Theorem \ref{T:MainTheorem} did not hold,
then the evolution equations given would be accurate to
$\cO(\epsilon)$. Whilst it would be ideal to have a proof that the
evolutions in Section \ref{S:epsilon1} do preserve the
non-negativity of $\psi$, as mentioned previously, this leads to
significant additional technical difficulties and may well require
further assumptions on the potentials and hydrodynamic interactions.


In addition, as mentioned in Section \ref{S:Introduction}, according
to DFT, $g$ is also a functional of $\rho$. If $\rho$ is independent
of $\epsilon$ (i.e.\ depends only on $f_0$) then the analysis is
unaffected. This is the case if the initial condition is independent
of $\epsilon$ and the terms $\psi_i$ (the part of $f_i$ in the null
space of $\tilde \cL$) are uniformly zero for all time and all $i$.
As mentioned earlier, such a result would rely on the dissipativity
of the determining equations, or equivalently on the equations being
non-negativity preserving.

If $\rho(\bsr,t)$ depends on $\epsilon$, then the nonlinear
operators are no longer quadratic in $f$ as $\rho$ depends on
higher-order, $\epsilon$-dependent parts of $f$.  Furthermore, we do
not know the precise dependence of $g$ on $\rho$.  However, if we
expand $g$ as a power series in $\epsilon$, $g=g_0 + \epsilon g_1 +
\dots$, then the equations in $\epsilon^{-2}$ and $\epsilon^{-1}$
change only by replacing $g$ with $g_0$.  This is because there are
no extra terms in the $\epsilon^{-2}$ equation, and only the extra
term $\cN(f_0,f_0,g_1)=0$ enters the $\epsilon^{-1}$ equation. The
main point is that $g_1$ enters only through non-linear terms,
namely by the addition of the terms $\cN_0(f_1,f_0,g_1)$,
$\cN_0(f_0,f_1,g_1)$ and $\cN_1(f_0,f_0,g_1)$ to the right hand side
of (\ref{eps0a}) (where we have now shown the explicit dependence of
the non-linear terms on the $g_j$). Thus the conclusion that the
dynamics of $f_0$ are governed by the solvability condition $\int
\dd \bsp (\cL_1 f_1 - \partial_t f_0)=0$ still holds.

The first difference comes when determining $f_2$, or more precisely $P_1
f_2$, which gains additional $g_1$-dependent terms.  The evolution equation for
$\psi$ looks superficially similar, but results in a new definition of $\tilde
\bsa_2$ and hence also of $\bsa_2$.


We note here that a similar argument applies if $g$ were chosen to depend
explicitly on time.  In particular, there are no further difficulties if $g$
depends only on the slowest timescale, i.e.\ if it is independent of $\epsilon$.
However, how one would choose this explicit dependence is unclear. The standard
approach is to choose $g$ to be either a functional of a suitably averaged
distribution $\bar\rho(t)$, or to satisfy the generalized Ornstein-Zernike
equation \cite{RexLowen09}.  In both cases, the time-dependence of $g$ is due
only to the time-dependence of $\rho$ and is not prescribed explicitly.


Allowing $g$ to depend (symmetrically) on $\bsp_1$ and $\bsp_2$ introduces many
additional complications in the analysis.  If the $\bsp$-dependence is
introduced at leading order it significantly changes the analysis of the
non-linear terms. Whilst, by (\ref{fnDefn}), it still holds that $\int \dd \bsr'
\dd \bsp' f^{(1)}(\bsr,\bsp,\tau) g(\bsr,\bsr',\bsp,\bsp')=N-1$, we actually
require an
expression for $\int \dd \bsp' f^{(1)}(\bsr,\bsp,\tau)
g(\bsr,\bsr',\bsp,\bsp')$.  For example, the $\cN_1(f_0,f_0)$ term in Lemma
\ref{L:L1N1} is significantly more complicated.


In this work we have solved the dynamics of the one-body
distribution up to errors of $\cO(\epsilon^2)$.  However, an
interesting question is whether the true solution and the solution
given by Theorem \ref{T:MainTheorem} are also close in some suitable
norm. Ideally, one would like to prove a result analogous to
Theorem 3.1 of \cite{EspositoLebowitzMarra99}, which states that,
under suitable assumptions, in a suitable norm, and for a fixed
macroscopic time period $(0,t_0]$, the solution to the Boltzmann
equation is $\mathcal{O}(\epsilon)$ close to the local Maxwellian
whose parameters vary according to the hydrodynamic equations. Note
that in the case where $t_0 \to \infty$, the constant in the
$\mathcal{O}(\epsilon)$ bound may diverge.

In order to prove such result, one must truncate the Hilbert series
at a finite order and add a remainder term.  One then determines
bounds on each of these terms, which require sufficiently good
estimates on the collision term (the hydrodynamic interactions and
$V_2$ in our case).  This truncation is necessary as the Hilbert
expansion does not converge uniformly in the small parameter.  Since
such estimates on the collision operator depend on its precise form
(in particular, it is assumed that the kernel of the Boltzmann
collision operator has finite range; not true for hydrodynamic
interactions), and a specific choice of norms, we have restricted
our analysis to determining the leading order terms in such an
expansion.

We close by discussing some open problems.  The first area concerns
confined fluids and the effects of boundaries.  Although the
external potential $V_1$ may be used to model boundaries which are
impermeable to the colloid particles but permeable to the fluid, a
truly confined fluid cannot be modelled in this way.  Extension to
such systems would require a treatment of the hydrodynamic
interactions caused by the boundaries.  Such effects break the
symmetry of the bath, as well as changing the mobility of the
colloid particles near the boundaries. Additional complications
would result from the presence of heterogeneities at boundaries,
which is indeed the case in practice. Heterogeneous boundaries,
either chemical or topographical, can have a significant effect on
the behavior of fluids both at both the micro-scale (e.g. they can
influence the thickness of the wetting layer in the immediate
vicinity of the boundary and corresponding wetting transitions) and
macro-scale (they can affect the shape of the gas-liquid interface
away from the
boundaries)~\cite{Elley98,Quere90,Parry08,Savva09,Savva10,Nold11}.
It would also be of interest to study mixtures of colloid particles,
e.g.\ a system with two types of particle which differ in their
sizes, masses, or interparticle potentials $V_2$. As mentioned
above, a full treatment of the problem would involve analysis of
boundary layer effects, including how the initial condition for the
Smoluchowski equation should be determined by that for the Kramers
equation.  These and related issues are currently under
investigation.

\section*{Acknowledgements}

We are grateful to Alexandr Malijevsk\'y, Andreas Nold and Peter
Yatsyshin for stimulating discussions on density functional theory.

\appendix
\section{Proofs of some lemmas of Section \ref{S:Hilbert}} \label{A:Proofs}

\subsection{Proof of Lemma \ref{L:IntegralPD}} \label{A:IntegralPD}
Let $\bsw=\big(\bsv(\bsr_1,\bsp_1,t), \dots \bsv(\bsr_N,\bsp_N,t)\big)^T$.
Since $\bGa(\bsr^N)$ is positive-definite, we have, for some $\delta>0$, $\bsw
\cdot \bGa \bsw \geq \delta \bsw \cdot \bsw$.  Hence
\begin{align*}
\bsw \cdot \bGa \bsw &= \sum_{i,j=1}^N \bsv(\bsr_i,\bsp_i,t)
\cdot \bGa_{ij}(\bsr^N) \bsv(\bsr_i,\bsp_i,t) \\
&= \sum_{i=1}^N \bsv(\bsr_i,\bsp_i,t) \cdot \Big[ \bsone + \sum_{j\neq i}
\bsZ_1(\bsr_i,\bsr_j) \Big] \bsv(\bsr_i,\bsp_i,t)  \\
&\qquad + \sum_{i \neq j}
\bsv(\bsr_i,\bsp_i,t) \cdot \bsZ_2(\bsr_i,\bsr_j) \bsv(\bsr_j,\bsp_j,t)\geq
\delta \sum_{i=1}^N |\bsv(\bsr_i,\bsp_i,t)|^2.
\end{align*}
Since $f^{(N)}(\bsr^N,\bsp^N,t)$ is non-negative, and positive on a set of
non-zero measure (as by definition $f^{(N)} \geq 0$ and $\int \dd \bsr^N \dd \bsp^N f^{(N)} = N$), we have
\begin{align*}
    & \int \dd \bsr^N \dd \bsp^N f^{(N)}(\bsr^N,\bsp^N,t) \Big( \sum_{i=1}^N
\bsv(\bsr_i,\bsp_i,t) \cdot \Big[ \bsone + \sum_{j\neq i}
\bsZ_1(\bsr_i,\bsr_j) \Big] \bsv(\bsr_i,\bsp_i,t)  \\
& \qquad \qquad \qquad \qquad \qquad \qquad \qquad + \sum_{i \neq j}
\bsv(\bsr_i,\bsp_i,t) \cdot \bsZ_2(\bsr_i,\bsr_j) \bsv(\bsr_j,\bsp_j,t) \Big)
\\
&\geq \delta \int \dd \bsr^N \dd \bsp^N f^{(N)}(\bsr^N,\bsp^N,t)
\sum_{i=1}^N |\bsv(\bsr_i,\bsp_i,t)|^2.
\end{align*}
By the symmetry of $f^{(N)}$, interchanging dummy variables of integration
gives
\begin{align*}
    & \int \dd \bsr^N \dd \bsp^N f^{(N)}(\bsr^N,\bsp^N,t) \Big( N
\bsv(\bsr_1,\bsp_1,t) \cdot \Big[ \bsone +
(N-1)\bsZ_1(\bsr_1,\bsr_2) \Big] \bsv(\bsr_1,\bsp_1,t)  \\
& \qquad \qquad \qquad \qquad \qquad \qquad \qquad + N(N-1)
\bsv(\bsr_1,\bsp_1,t) \cdot \bsZ_2(\bsr_1,\bsr_2) \bsv(\bsr_2,\bsp_2,t) \Big)\\
&\geq \delta N \int \dd \bsr^N \dd \bsp^N f^{(N)}(\bsr^N,\bsp^N,t)
|\bsv(\bsr_1,\bsp_1,t)|^2.
\end{align*}
Using (\ref{fnDefn}) for the cases with $n=1$ and $n=2$, i.e.\
\begin{align*}
f^{(2)}(\bsr_1,\bsp_1,\bsr_2,\bsp_2,t) &= N(N-1)\int \dd \bsr^{N-2} \dd
\bsp^{N-2} f^{(N)}(\bsr^N,\bsp^N,t) \mbox{ and}\\
f^{(1)}(\bsr_1,\bsp_1,t) &= N\int \dd \bsr^{N-1} \dd
\bsp^{N-1} f^{(N)}(\bsr^N,\bsp^N,t)
\end{align*}
gives
\begin{align*}
    & \int \dd \bsr_1 \dd \bsp_1 \dd \bsr_2 \dd \bsp_2
f^{(2)}(\bsr_1,\bsp_1,\bsr_2,\bsp_2,t) \Big( \bsv(\bsr_1,\bsp_1,t) \cdot \Big[
\tfrac{1}{N-1} \bsone +
\bsZ_1(\bsr_1,\bsr_2) \Big] \bsv(\bsr_1,\bsp_1,t)  \\
& \qquad \qquad \qquad \qquad \qquad \qquad \qquad \qquad \qquad \qquad +
\bsv(\bsr_1,\bsp_1,t) \cdot \bsZ_2(\bsr_1,\bsr_2) \bsv(\bsr_2,\bsp_2,t) \Big)
\\
&\geq \delta \int \dd \bsr_1 \dd \bsp_1 f(\bsr_1,\bsp_1,t)
|\bsv(\bsr_1,\bsp_1,t)|^2.
\end{align*}
Inserting the definition
$f^{(2)}(\bsr_1,\bsp_1,\bsr_2,\bsp_2,t)=g(\bsr_1,\bsr_2)f(\bsr_1,\bsp_1,
t)f(\bsr_2,\bsp_2,t)$ and renaming the dummy variables gives the result.

The fact that $f$ may be chosen as $f^{(1)}$ is trivial.  To see that the result
holds when $f$ is replaced by $f_0$, we insert the expansion (\ref{fexpansion})
and note $\int \dd \bsr' \dd \bsp' g(\bsr,\bsr') f^{(1)}(\bsr',\bsp',t) = N-1$
holds for all $\epsilon$, in particular for $\epsilon = 0$.
$\endproof$

\subsection{Proof of Lemma \ref{L:SelfAdjoint}} \label{A:SelfAdjoint}
We consider each of the three operators in $\tilde \cL$ individually, starting
with $\cL_0$.  For arbitrary $f$, $\tilde f$, and using Corollary \ref{C:f0},
in particular that $\grad{\bsp} (f_0^{-1}) = \bsp f_0^{-1}$
\begin{align*}
    \langle f, \cL_0 \tilde f \rangle_{f_0^{-1}} &= \int \dd \bsr \dd \bsp
f_0^{-1}(\bsr,\bsp,t) f(\bsr,\bsp,t)
\div{\bsp}[(\bsp + \grad{\bsp}) \tilde f(\bsr,\bsp,t) ] \\
&= - \int \dd \bsr \dd \bsp f_0^{-1}(\bsr,\bsp,t) (\bsp + \grad{\bsp})
f(\bsr,\bsp,t)  \cdot
[(\bsp + \grad{\bsp}) \tilde f(\bsr,\bsp,t) ] \\
& = \int \dd \bsr \dd \bsp f_0^{-1}(\bsr,\bsp,t) \tilde f(\bsr,\bsp,t)
\div{\bsp}[ (\bsp + \grad{\bsp}) f(\bsr,\bsp,t)] ,
\end{align*}
where the second and third lines both follow via integration by parts.  Hence
$\cL_0$ is self-adjoint.

For $\cN_0(f_0,\tilde f)$, by Lemma \ref{L:N0} we have
\begin{align*}
    &\langle f, \cN_0(f_0,\tilde f)\rangle_{f_0^{-1}} = \int \dd \bsr \dd \bsp
f_0^{-1}(\bsr,\bsp,t)f(
\bsr,\bsp,t) \\
&\qquad\qquad\qquad\qquad \times\div{\bsp} \Big[ \int \dd \bsr' \dd \bsp'
g(\bsr,\bsr') f_0
(\bsr',\bsp',t)
\bsZ_1(\bsr,\bsr')  \times (\bsp + \grad{\bsp}) \tilde
f(\bsr,\bsp,t) \Big]\\
&=-\int \dd \bsr \dd \bsp f_0^{-1}(\bsr,\bsp,t) (\bsp+\grad{\bsp})
f(\bsr,\bsp,t) \\
&\qquad\qquad\qquad\qquad\cdot \Big[ \int \dd \bsr'
\dd \bsp' g(\bsr,\bsr')
 f_0 (\bsr',\bsp',t) \bsZ_1(\bsr,\bsr') \times (\bsp + \grad{\bsp}) \tilde
f(\bsr,\bsp,t) \Big]\\
&= -\int \dd \bsr \dd \bsp \dd \bsr' \dd \bsp' f_0^{-1}(\bsr,\bsp,t)
g(\bsr,\bsr') f_0(\bsr',\bsp',t) \\
&\qquad\qquad\qquad\qquad \times (\bsp + \grad{\bsp}) f(\bsr,\bsp,t) \cdot
\bsZ_1(\bsr,\bsr')(\bsp +
\grad{\bsp})
\tilde f(\bsr,\bsp,t),
\end{align*}
where we have used integration by parts and Fubini's theorem.  Now, since
$\bsZ_1$ is a symmetric matrix, $f$ and $\tilde f$ can be interchanged and the
argument reversed, showing that $\cN_0(f_0,\tilde f)$ is self-adjoint.

It remains to calculate the adjoint of $\cN_0(\tilde f,f_0)$.  Using Lemma
\ref{L:N0} gives
\begin{align*}
    &\langle f, \cN_0(\tilde f,f_0)\rangle_{f_0^{-1}} = - \int \dd \bsr \dd \bsp
f_0^{-1}(\bsr,\bsp,t)  f(\bsr,\bsp,t) f_0(\bsr,\bsp,t)\\
&\qquad \qquad\qquad\qquad\qquad \times  \int \dd \bsr' \dd
\bsp' g(\bsr,\bsr') \bsZ_2(\bsr,\bsr') (\bsp' + \grad{\bsp'}) \tilde
f(\bsr',\bsp',t) \cdot \bsp\\
&= - \int \dd \bsr \dd \bsp \dd \bsr' \dd \bsp'
g(\bsr,\bsr') f(\bsr,\bsp,t)   \tilde f(\bsr',\bsp',t) \bsp
\cdot \bsZ_2(\bsr,\bsr') \bsp' ,
\end{align*}
where we have used the divergence theorem, Fubini's theorem and the identity
(for symmetric matrices) $\bsZ_2 \bsp' \cdot \bsp = \bsp \cdot \bsZ_2 \bsp'$.
Since this final term, along with the rest of the integral is symmetric under
interchanging the pairs of dummy variables $(\bsr,\bsp) \leftrightarrow
(\bsr',\bsp')$ we see that  $\cN_0(\tilde f,f_0)$ is also self-adjoint.
The overall result now follows from linearity of the integral, and
hence of the adjoint.
$\endproof$

\subsection{Proof of Lemma \ref{L:resolvent}} \label{A:resolvent}
We prove the equivalent statement (for self adjoint operators, as $\tilde \cL$
is by Lemma \ref{L:SelfAdjoint}) that there exists an orthonormal basis
$(\xi_j)_{j=1}^\infty$ of $L^2_{f_0^{-1}}$ such that $\tilde \cL \xi_j =
\lambda_j \xi_j$, $\lambda_j \in \R$ such that $\lim_{j\to \infty} |\lambda_j| =
\infty$ \cite[Theorem 11.3.13]{deOliveria09}.  We make use of Lemma
\ref{L:IntegralPD}, which allows us to compare the eigenvalues of $\tilde L$ to
those where $\bsZ_i\equiv 0$, and the fact that the eigenfunctions and
eigenvalues of
the resulting operator can be constructed explicitly.  We note that the
$P_{n,\bsa}$ form a basis of $L^2(\R^3, \e{-|\bsp|^2/2})$ and so the functions
$\e{-|\bsp|^2/2}P_{n,\bsa}$ form a basis of $L^2(\R^3, \e{|\bsp|^2/2})$.

First note, by Lemma \ref{L:preservesN}, that the spaces
\[
\mbox{Span}\{ \e{-|\bsp|^2/2} P_{n,\bsa}(\bsp) | n \mbox{ fixed},
|\bsa|=n\}
\]
(where the coefficients may be functions of $\bsr$, $t$) are invariant under
$\tilde \cL$.  Thus all eigenfunctions may be written in
the form
\be
    \label{efnForm}
    \psi_{n,j}(\bsr,\bsp,t)=\sum_\bsa \beta_{\bsa,j}(\bsr,t) \e{-|\bsp|^2/2}
P_{n,\bsa}(\bsp),
\ee
where $j=1, \dots, T(n+1)$, with $T(n)$ the $n$-th triangular number (which
corresponds the the number of solutions to $a_1+a_2+a_3=n-1$).
Also, as noted in the proof of Lemma \ref{L:preservesN}, $\tilde
\cN(f,f_0)$ contributes only for $n=1$, and as such we may ignore it when
calculating the eigenvalues.  It therefore suffices to consider the eigenvalues
of
\begin{align*}
    \bar \cL f(\bsr,\bsp,t) &= \div{\bsp} \Big[ \int \dd \bsr' \dd \bsp'
g(\bsr,\bsr') f_0(\bsr',\bsp',t) \big( \tfrac{1}{N-1}\bsone + \bsZ_1(\bsr,\bsr')
\big) (\bsp+\grad{\bsp}) f(\bsr,\bsp,t) \Big] \\
&=: \div{\bsp} \bar\bsZ(\bsr,t)(\bsp+\grad{\bsp})f(\bsr,\bsp,t).
\end{align*}

Suppose $-\bar \cL \psi_{n,j} = \lambda_{n,j} \psi_{n,j}$, then
\begin{align*}
&\lambda_{n,j} \int \dd \bsr \dd \bsp f_0^{-1}(\bsr,\bsp,t)
|\psi_{n,j}(\bsr,\bsp,t)|^2\\
&\qquad = -\int \dd \bsr \dd \bsp f_0^{-1}(\bsr,\bsp,t) \psi_{n,j}(\bsr,\bsp,t)
\div{\bsp} \bar\bsZ(\bsr,t)(\bsp+\grad{\bsp}) \psi_{n,j}(\bsr,\bsp,t)\\
&\qquad = \int \dd \bsr \dd \bsp f_0^{-1}(\bsr,\bsp,t)
(\bsp+\grad{\bsp})\psi_{n,j}(\bsr,\bsp,t) \cdot
\bar\bsZ(\bsr,t)(\bsp+\grad{\bsp}) \psi_{n,j}(\bsr,\bsp,t)\\
&\qquad=\int \dd \bsr \dd \bsp \dd \bsr' \dd \bsp' f_0^{-1}(\bsr,\bsp,t)
g(\bsr,\bsr') f_0(\bsr',\bsp',t)\\
& \qquad \qquad \qquad
\times (\bsp+\grad{\bsp})\psi_{n,j}(\bsr,\bsp,t) \cdot
\big( \tfrac{1}{N-1}\bsone + \bsZ_1(\bsr,\bsr')
\big)(\bsp+\grad{\bsp}) \psi_{n,j}(\bsr,\bsp,t)\\
&\qquad =\int \dd \bsr \dd \bsp \dd \bsr' \dd \bsp'
f_0(\bsr,\bsp,t) f_0(\bsr',\bsp',t) g(\bsr,\bsr')
\\
& \qquad \qquad \qquad
\times \bsv(\bsr,\bsp,t) \cdot
\big( \tfrac{1}{N-1}\bsone + \bsZ_1(\bsr,\bsr')
\big)\bsv(\bsr,\bsp,t),
\end{align*}
where the second equality follows via integration by
parts and that
$f_0^{-1}(\bsr,\bsp,t)=\rho_0^{-1}(\bsr,t)Z\exp(|\bsp|^2/2)$, and we denote
$\bsv(\bsr,\bsp,t) = (\bsp+\grad{\bsp})\psi_{n,j}(\bsr,\bsp,t)
f_0^{-1}(\bsr,\bsp,t)$. Now
note that Lemma \ref{L:IntegralPD} holds when $\bsZ_2$ is set to zero since it
requires only that $\bGa$ is positive definite with the correct symmetry.  Since
$\bGa_{11}$ is a principal minor of $\bGa$, it is positive definite, and by
symmetry so are all $\bGa_{jj}$.  It therefore follows that the block diagonal
matrix with entries $\bGa_{jj}$ is also positive definite, with the same required
symmetry as $\bGa$ and we have
\begin{align*}
&\lambda_{n,j} \int \dd \bsr \dd \bsp  f_0^{-1}(\bsr,\bsp,t)
|\psi_{n,j}(\bsr,\bsp,t)|^2 \\
&\qquad \geq \delta \int \dd \bsr \dd \bsp
f_0(\bsr,\bsp,t)|\bsv(\bsr,\bsp,t)|^2 = \delta \int \dd \bsr \dd \bsp
f_0^{-1}(\bsr,\bsp,t) |(\bsp+\grad{\bsp})\psi_{n,j}(\bsr,\bsp,t)|^2\\
&\qquad= -\delta \int \dd \bsr \dd \bsp f_0^{-1}(\bsr,\bsp,t)
\psi_{n,j}(\bsr,\bsp,t) \div{\bsp} (\bsp+\grad{\bsp})\psi_{n,j}(\bsr,\bsp,t).
\end{align*}

We now compute $\div{\bsp} (\bsp+\grad{\bsp})\psi_{n,j}(\bsr,\bsp,t)$:
\begin{align*}
    &\div{\bsp} (\bsp + \grad{\bsp}) \e{-|\bsp|^2/2} P_{n,\bsa}(\bsp)  =
\sum_{j=1}^3 \partial_{p_j}(p_j + \partial_{p_j})\e{-|\bsp|^2/2}
P_{n,\bsa}(\bsp) \\
&\qquad = \sum_{j=1}^3 \partial_{p_j} \big( \e{|\bsp|^2/2} \partial_{p_j} P_{n,
\bsa}(\bsp)\big)
= \sum_{j=1}^3 \partial_{p_j} \big( \e{|\bsp|^2/2} a_j P_{n,
\bsa-\bse_j}(\bsp)\big) \\
&\qquad = \sum_{j=1}^3 \e{|\bsp|^2/2} a_j (-p_j + \partial_{p_j}) P_{n,
\bsa-\bse_j}(\bsp) =- \e{|\bsp|^2/2}P_{n, \bsa}(\bsp) \sum_{j=1}^3 a_j
\\ & \qquad = -n \e{|\bsp|^2/2}P_{n, \bsa}(\bsp).
\end{align*}
The required identities for operators on $P_{n, \bsa}$ follow from its product
form and the equivalent 1-dimensional identities.  Note $\bse_j$ is the $j$th
unit vector.  Thus we have
\[
\lambda_{n,j} \int \dd \bsr \dd \bsp  f_0^{-1}(\bsr,\bsp,t)
|\psi_{n,j}(\bsr,\bsp,t)|^2 \geq n \delta \int \dd \bsr \dd \bsp
f_0^{-1}(\bsr,\bsp,t) |\psi_{n,j}(\bsr,\bsp,t)|^2
\]
and the result follows.
$\endproof$

\subsection{Proof of Lemma \ref{L:preservesN}} \label{A:preservesN}
From (\ref{L0}) and Lemma \ref{L:N0}, for $\cL_0 f$ and $\cN_0(f_0,f)$, it is
sufficient to consider a general operator
\[
    \overline \cL:=\grad{\bsp} \cdot \bsZ(\bsr) (\bsp+ \grad{\bsp}) =
\sum_{i,j} Z_{ij} \partial_{p_i}(p_j + \partial_{p_j}).
\]
We have, using the standard identity $\partial_x H_n(x)=n H_{n-1}(x)$,
\begin{align}
    (p + \partial_{p}) [H_{a}(p)\e{-p^2/2}] &=
    \e{-p^2/2} \partial_{p} H_a(p) = \e{-p^2/2} a H_{a-1}(p)
    \label{p+dp}
\end{align}
and using $H_{n+1}(x)=xH_n(x)-\partial_x H_n(x)$, we find
\begin{align*}
    \partial_{p}[H_{a}(p)\e{-p^2/2}] & = \e{-p^2/2}(-p+\partial_{p})H_{a}(p)
    = - H_{a+1}(p)\e{-p^2/2}
\end{align*}
It is therefore clear that $\overline \cL$ preserves $|\bsa|=n$, with the
possibility of the new coefficients all being zero.

It remains to consider $\cN_0(f,f_0)$, which by Lemma \ref{L:N0} is given by
\[
\cN_0(f,f_0)= - \tfrac{1}{mk_BT} f_0(\bsr,\bsp,t) \int \dd \bsr'
\dd \bsp' g(\bsr,\bsr') \bsZ_2(\bsr,\bsr') (\bsp' + \grad{\bsp'})
f(\bsr',\bsp',t) \cdot \bsp
\]
Note that $P_0=1$.  Using (\ref{p+dp}), $f_0 = Z^{-1} \e{-p^2/(2mk_BT)}
\rho_0(\bsr,t)$, and that the $P_{n,\bsa}$ are orthogonal, it is clear that the
integral gives zero for any terms not proportional to $p_i=H_{1,i}$, and in
this case returns something of the form $\alpha(\bsr)\cdot \bsp$.

Hence $\tilde \cL$ preserves $n$ and it remains to show that $\tilde \cL
\sum_{|\bsa|=n} \alpha_{n,\bsa}(\bsr,p)P_{n,\bsa}(\bsp) = 0$ if and only if
$\alpha_{n,\bsa}=0$ for all $\bsa$.  This follows from the null space of
$\tilde \cL$ being  $\e{-|\bsp|^2/2} \phi(\bsr,t)$ (see Lemma
\ref{L:nullLtildestar}), and thus containing only $P_0$, and the orthogonality
of the $P_{n,\bsa}$.
$\endproof$

\bibliographystyle{siam}
\bibliography{84465}
\end{document}